\documentclass[11pt,a4paper,reqno]{amsart}%
\usepackage{amsthm,amsmath,amsfonts,amssymb,amsxtra,appendix,bookmark,dsfont,bm,mathrsfs}

\theoremstyle{plain}
\newtheorem{theorem}{Theorem}
\newtheorem{lemma}[theorem]{Lemma}

\theoremstyle{definition}

\theoremstyle{remark}
\newtheorem{remark}[theorem]{Remark}

\def\leqslant{\le}
\def\bq{\begin{eqnarray}}
\def\eq{\end{eqnarray}}
\def\bqq{\begin{align*}}
\def\eqq{\end{align*}}

\def\nn{\nonumber}

\def\eps{\varepsilon}

\renewcommand{\epsilon}{\varepsilon}
\newcommand\1{{\ensuremath {\mathds 1} }}

\def\cF {\mathcal{F}}

\def\cN{\mathcal{N}}

\def\R {\mathbb{R}}

\def\C {\mathbb{C}}

\def\H{\gH}

\def\R {\mathbb{R}}
\def\C {\mathbb{C}}

\def\d{\,{\rm d}}
\newcommand{\gH}{\mathfrak{H}}

\newcommand{\bH}{\mathbb{H}}
\newcommand{\dGamma}{{\ensuremath{\rm d}\Gamma}}

\allowdisplaybreaks

\title[Bogoliubov correction to mean-field dynamics]{A note on the validity of Bogoliubov correction to mean-field dynamics}

\author[P.T. Nam]{Phan Th\`anh Nam}
\address{Institute of Science and Technology Austria, Am Campus 1, 3400 Klosterneuburg, Austria} 
\email{pnam@ist.ac.at}
\author[M. Napi\'orkowski]{Marcin Napi\'orkowski}
\address{Institute of Science and Technology Austria, Am Campus 1, 3400 Klosterneuburg, Austria} 
\email{mnapiork@ist.ac.at}

\begin{document}
\date{\today}

\begin{abstract} We study the norm approximation to the Schr\"odinger dynamics of $N$ bosons in $\R^3$ with an interaction potential of the form $N^{3\beta-1}w(N^{\beta}(x-y))$. Assuming that in the initial state the particles outside of the condensate form a quasi-free state with finite kinetic energy, we show that in the large $N$ limit, the fluctuations around the condensate can be effectively described using Bogoliubov approximation for all $0\le \beta<1/2$. The range of $\beta$ is expected to be optimal for this large class of initial states. 
\end{abstract}

\maketitle


\section{Introduction}

We are interested in the norm approximation of the Schr\"odinger evolution 
\bq \label{eq:schrodingerdynamics}
\Psi_N(t) = e^{-itH_N}\Psi_{N}(0)
\eq
on the bosonic Hilbert space $\H^N=\bigotimes_{\text{sym}}^N L^2(\R^3)$. Here 
\begin{equation*} 
H_N= \sum\limits_{j = 1}^N -\Delta_{x_j} + \frac{1}{N-1} \sum\limits_{1 \leqslant j < k \leqslant N} {w_N(x_j-x_k)}
\end{equation*}
is the Hamiltonian of a system of $N$ identical bosons in $\R^3$. The interaction potential is taken of the delta-type form 
\begin{equation*}  
w_N(x-y)= N^{3\beta} w(N^\beta (x-y)).
\end{equation*}
The parameter $\beta \ge 0$ is fixed and $w\in C^1_0(\R^3)$ is non-negative and spherically symmetric decreasing. 

From the physical point of view, the initial state $\Psi_N(0)$ may be interpreted as a ground state of a trapped system and the time evolution $\Psi_N(t)$ in \eqref{eq:schrodingerdynamics} is observed when the trapping potential is turned off. Thus, motivated by the results on ground states in  \cite{LewNamSerSol-15} (see also \cite{Seiringer-11,GreSei-13,DerNap-13,NamSei-15}), we expect that
\bq \label{eq:PsiN0-intro}
\Psi_N(0) \approx \sum_{n=0}^N u(0)^{\otimes (N-n)} \otimes_s \psi_n(0)
\eq
in norm for $N$ large. Here $u(0)$ is a normalized function in $L^2(\R^3)$ describing the Bose-Einstein condensate and $(\psi_n(0))_{n=0}^\infty$ is a {quasi-free} state describing the fluctuations around the condensate. 

We will show that if  $\Psi_N(0)$ satisfies \eqref{eq:PsiN0-intro}, then for every $t>0$, we have
\bq \label{eq:PsiNt-intro}
\Psi_N(t) \approx \sum_{n=0}^N u(t)^{\otimes (N-n)} \otimes_s \psi_n(t)
\eq
in norm for $N$ large. Here $u(t) \in L^2(\R^3)$ is determined by a mean-field (Hartree) equation and $(\psi_n(t))_{n=0}^\infty$ is a quasi-free state governed by a quadratic (Bogoliubov) Hamiltonian on Fock space. 

The approximation \eqref{eq:PsiNt-intro} has been first established in \cite{LewNamSch-15} for $\beta=0$, and then extended to $0\le \beta<1/3$ in \cite{NamNap-15}. This range of $\beta$ seems to be optimal if we only assume that $(\psi_n(0))_{n=0}^\infty$ in \eqref{eq:PsiN0-intro} is a quasi-free state. In the present work, we will make an additional assumption (still physically reasonable) that $(\psi_n(0))_{n=0}^\infty$ has finite kinetic energy, and prove \eqref{eq:PsiNt-intro} for all $0\le \beta<1/2$. Note that when $\beta>1/3$, the range of the interaction potential is much smaller than the average distance between the particles, and hence every particle essentially interacts only with itself. This so-called self-interaction regime is physically more relevant and mathematically more challenging than the mean-field regime $\beta <1/3$. 

An analogue of \eqref{eq:PsiNt-intro} related to the fluctuations around coherent states in Fock space has been justified in \cite{Hepp-74,GinVel-79, GinVel-79b,GriMacMar-10,GriMacMar-11} for $\beta=0$, in \cite{GriMac-13} for $\beta<1/3$ and in  \cite{Kuz-15b} for $\beta<1/2$.  In particular, our result is comparable to \cite{Kuz-15b}, but our method is different and it can be used to simplify the proof in \cite{Kuz-15b}. Thanks to a heuristic argument in \cite{Kuz-15b}, we also expect that the range $0\le \beta<1/2$ is optimal for the approximation \eqref{eq:PsiNt-intro} to hold, as soon as $u(t)$ is still decoupled from the equation for $(\psi_n(t))_{n=0}^\infty$.   

When $\beta>1/2$, the effective equations for $u(t)$ and $(\psi_n(t))_{n=0}^\infty$ in \eqref{eq:PsiNt-intro} have to be modified to take two-body scattering processes into account. This step has been carried out in the Fock space setting in \cite{BocCenSch-15,GriMac-15}, but it is still open in the $N$-particle setting. 

Note that the norm convergence \eqref{eq:PsiNt-intro} is much more precise than the usual convergence of density matrices in the context of the Bose-Einstein condensation. In particular, our result can be interpreted as a second order correction to the leading order result in \cite{ErdSchYau-07}. We refer to \cite{NamNap-15} for a further discussion and an extended list of literature in this direction.  

The precise statement of our result is given in the next section.

\medskip

\noindent{\bf Acknowledgment.}  The first author thanks Mathieu Lewin for helpful discussions. We thank David Mitrouskas and S\"oren Petrat for finding a gap in a previous version of this paper. We thank the referees for useful comments and remarks. The support of the Austrian Science Fund (FWF) project Nr. P 27533-N27 is gratefully acknowledged.

\section{Main result} \label{sec:main-result}

In our paper, the condensate is governed by the Hartree equation
\begin{equation} \label{eq:Hartree-equation}
\left\{
\begin{aligned}
i\partial_t u(t) &=  \big(-\Delta +w_N*|u(t)|^2 -\mu_N(t)\big) u(t), \\ 
u(t=0)&=u(0).
\end{aligned}
\right.
\end{equation}
Here we can choose the phase
$$
\mu_N(t)=\frac12\iint_{\R^3\times\R^3}|u(t,x)|^2w_N(x-y)|u(t,y)|^2 \d x \d y
$$
to ensure an energy compatibility (see \cite{LewNamSch-15} for further explanations). The well-posedness of the Hartree equation is recalled in Lemma \ref{lem:Hartree-equation}.

To describe the fluctuations around the condensate, it is natural to introduce the Fock space
$$ \cF(\gH)= \bigoplus_{n=0}^\infty \gH^n= \C \oplus \bigoplus_{n=1}^\infty \bigotimes^n_{\rm sym} \gH, \quad \gH=L^2(\R^3).$$ 
On this Fock space, we define the creation and annihilation operators $a^*(f)$, $a(f)$, with $f\in \gH$, by
\begin{align*}
(a^* (f) \Psi )(x_1,\dots,x_{n+1})&= \frac{1}{\sqrt{n+1}} \sum_{j=1}^{n+1} f(x_j)\Psi(x_1,\dots,x_{j-1},x_{j+1},\dots, x_{n+1}), \\
 (a(f) \Psi )(x_1,\dots,x_{n-1}) &= \sqrt{n} \int \overline{f(x_n)}\Psi(x_1,\dots,x_n) \d x_n, \quad \forall \Psi \in \gH^n,\, \forall n. 
\end{align*}
These operators satisfy the canonical commutation relations (CCR)
$$ 
[a(f),a(g)]=[a^*(f),a^*(g)]=0,\quad [a(f), a^* (g)]= \langle f, g \rangle, \quad \forall f,g \in \gH.
$$
Equivalently, we can define the operator-valued distributions $a_x^*$ and $a_x$, with $x\in\R^3$, by
$$
a^*(f)=\int_{\R^3}  f(x) a_x^* \d x, \quad a(f)=\int_{\R^3} \overline{f(x)} a_x \d x, \quad \forall f\in \gH.
$$
They satisfy  
$$[a^*_x,a^*_y]=[a_x,a_y]=0, \quad [a_x,a^*_y]=\delta(x-y), \quad \forall x,y\in \R^3.$$
These operators allow us to express operators on Fock space in a convenient way. For example, for every operator $h$ on $L^2(\R^3)$ with kernel $h(x,y)$, we can write
$$
\dGamma(h):= 0\oplus \bigoplus_{n=0}^\infty \sum_{j=1}^n h_j = \int_{\R^3} a_x^* h a_x \d x = \iint_{\R^3\times \R^3} h(x,y) a_x^* a_y \d x \d y.
$$
In particular, $\cN=\dGamma(1)$ is called the number operator.

In our approximation \eqref{eq:PsiN0-intro}-\eqref{eq:PsiNt-intro}, the particles outside of the condensate are described by a unit vector $\Phi(t)=(\psi_n(t))_{n=0}^\infty$ in the excited Fock space
$$ \cF_+(t)= \bigoplus_{n=0}^\infty \bigotimes^n_{\rm sym} \gH_+(t), \quad \gH_+(t)=\{u(t)\}^\bot=Q(t)\gH, \quad Q(t):=1-|u(t) \rangle \langle u(t)|.$$
This vector is governed by the Bogoliubov equation
\bq \label{eq:Bogoliubov-equation}
\left\{
\begin{aligned}
i \partial_t \Phi(t) &=  \bH(t) \Phi(t),\\
\Phi(t=0)&= \Phi(0),
\end{aligned}
\right.
\eq
where 
\begin{align*} 
&\bH(t):= \dGamma(h(t)) + \frac12\iint_{\R^3\times\R^3}\Big(K_2(t,x,y)a^*_x a^*_y +\overline{K_2(t,x,y)}a_x a_y\Big)\d x\,\d y, \\
& h(t)=-\Delta+|u(t,\cdot)|^2\ast w_N -\mu_N(t) + Q(t) \widetilde{K}_1(t) Q(t), \\
& K_2(t, \cdot, \cdot)=Q(t)\otimes Q(t)\widetilde{K}_2(t, \cdot, \cdot). 
\end{align*}
Here $\widetilde{K}_1(t)$ is the operator on $\gH$ with kernel $\widetilde{K}_1(t,x,y)=u(t,x)w_N(x-y)\overline{u(t,y)}$, and $\widetilde{K}_2(t,x,y)=u(t,x)w_N(x-y)u(t,y)$. A heuristic derivation of  \eqref{eq:Bogoliubov-equation} will be revised in Section \ref{sec:Bogoliubov}. 

We will restrict our attention to quasi-free states. Recall that a unit vector $\Psi\in \cF(\gH)$ is called a quasi-free state if it has finite particle number expectation, namely $\langle \Psi, \cN \Psi \rangle<\infty$, and satisfies Wick's Theorem: 
\begin{align*}
&\langle \Psi, a^{\#}(f_{1}) a^{\#}(f_{2}) \cdots a^{\#}(f_{2n-1}) \Psi \rangle = 0,  \\
&\langle \Psi, a^{\#}(f_{1}) a^{\#}(f_{2}) \cdots a^{\#}(f_{2n})  \Psi \rangle = \sum_{\sigma} \prod_{j=1}^n \langle \Psi, a^{\#}(f_{\sigma(2j-1)}) a^{\#}(f_{\sigma(2j)}) \Psi \rangle 
\end{align*}
for all $f_1,...,f_n \in \gH$ and for all $n$. Here $a^{\#}$ is either the creation or annihilation operator and the sum is taken over all permutations $\sigma$ satisfying $\sigma(2j-1)<\min\{\sigma(2j),\sigma(2j+1) \}$ for all $j$. By the definition, any quasi-free state is determined uniquely (up to a phase) by its one-body density matrices $\gamma_\Psi: \gH\to \gH$ and $\alpha_\Psi:\overline{\gH} \equiv \gH^* \to {\gH}$ which are defined by
$$
\left\langle {f,{\gamma _\Psi }g} \right\rangle  = \left\langle \Psi, {{a^*}(g)a(f)} \Psi \right\rangle,\quad \left\langle {{f}, \alpha _\Psi \overline{g} } \right\rangle  = \left\langle \Psi, {a(g)a(f)} \Psi\right\rangle, \quad \forall f,g \in \gH.
$$

In \cite{NamNap-15}, we proved that if $\Phi(0)$ is a quasi-free state, then the solution $\Phi(t)$ to \eqref{eq:Bogoliubov-equation} is a quasi-free state for all $t>0$ and $(\gamma_{\Phi(t)}, \alpha_{\Phi(t)})$ is the unique solution to the system
\bq \label{eq:linear-Bog-dm} 
\left\{
\begin{aligned}
i\partial_t \gamma &= h \gamma - \gamma h + K_2 \alpha - \alpha^* K_2^*, \\
i\partial_t \alpha &= h \alpha + \alpha h^{\rm T} + K_2  + K_2 \gamma^{\rm T} + \gamma K_2,\\
\gamma(t&=0)=\gamma_{\Phi(0)}, \quad \alpha(t=0)  = \alpha_{\Phi(0)}.
\end{aligned}
\right.
\eq
Here $K_2$ is interpreted as an operator $\gH^*\to \gH$ with kernel $K_2(t,x,y)$. Note that \eqref{eq:linear-Bog-dm} is similar (but not identical) to the equations studied in \cite{GriMac-13,Kuz-15b,BacBreCheFroSig-15}.  The well-posedness of \eqref{eq:Bogoliubov-equation}-\eqref{eq:linear-Bog-dm} is recalled in Lemma \ref{lem:Bogoliubov-equation}. 

Now we are ready to state our main result.

\begin{theorem}[Validity of Bogoliubov dynamics] \label{thm:main}  Let $0\le \beta<1/2$.
\begin{itemize}
\item Let $u(t)$ satisfy the Hartree equation \eqref{eq:Hartree-equation}, where the (possibly $N$-dependent) initial state $u(0,\cdot)$ satisfies 
$$\| u(0,\cdot)\|_{W^{\ell,1}(\R^3)} \le \kappa_0$$
for $\ell$ sufficiently large and for a constant $\kappa_0>0$ independent of $N$.

\medskip

\item Let $\Phi(t)=(\psi_n(t))_{n=0}^\infty \in \cF_+(t)$ satisfy the Bogoliubov equation \eqref{eq:Bogoliubov-equation}, where the (possibly $N$-dependent) initial state $\Phi(0)$ is a quasi-free state in $\cF_+(0)$ satisfying
\bq \label{eq:assumption-Phi0}
\big\langle \Phi(0), \cN  \Phi(0) \big\rangle \le \kappa_\eps N^{\eps}\quad \text{and} \quad \big\langle \Phi(0), \dGamma(1-\Delta) \Phi(0) \big\rangle\le \kappa_\eps N^{\beta+\eps}
\eq 
for all $\eps>0$, where the constant $\kappa_\eps>0$ is independent of $N$.

\medskip

\item Let $\Psi_N(t)$ satisfy the Schr\"odinger equation \eqref{eq:schrodingerdynamics} with the initial state
\bq \label{eq:PhiN0-thm}
\Psi_N(0) = \sum_{n=0}^N u(0)^{\otimes (N-n)} \otimes_s \psi_n(0) = \sum_{n=0}^N \frac{(a^*(u(0)))^{N-n}}{\sqrt{(N-n)!}} \psi_n(0). 
\eq
\end{itemize}
Then for all $\eps>0$ and for all $t>0$ we have 
\begin{align} \label{eq:thm-mainresult}
\Big\| \Psi_N(t) - \sum_{n=0}^N u(t)^{\otimes (N-n)} \otimes_s \psi_n(t) \Big\|_{\gH^N}^2 \le C_\eps (1+t)^{1+\eps} N^{(2\beta-1+\eps)/2}
\end{align}
where the constant $C_\eps>0$ depends only on $\kappa_0$ and $\eps$.
\end{theorem}

Strictly speaking, the initial state $\Psi_N(0)$ is not normalized. However, its norm converges to $1$ very fast when $N\to \infty$ (we will see it from the proof). We ignore this trivial normalization in the statement of Theorem \ref{thm:main} for simplicity. 

Since $\Psi_N(0)$ is expected to be the ground state of a trapped system with the interaction potential $w_N(x-y)$, the initial data $u(0,\cdot)$ and $\Phi(0)$  are allowed to depend on $N$. In particular, the assumptions \eqref{eq:assumption-Phi0} on $\Phi(0)$ are motivated by the ground state properties of quadratic Hamiltonians (see Remark \ref{rmk:ground-state}). More generally, we can also assume that \eqref{eq:assumption-Phi0} holds for {\em some} $\eps>0$, and replace the right side of \eqref{eq:thm-mainresult} by $C_\eps (1+t)^{1+\eps} N^{(2\beta-1+9\eps)/2}$ (see the estimate \eqref{eq:thm-quantitative-estimate} in the proof).

Our proof builds on ideas in \cite{LewNamSch-15,NamNap-15}, where the case $0\le \beta<1/3$ was studied. However, the extension to $\beta<1/2$ requires several new tools, most notably a new kinetic estimate for the particles outside of the condensate (see Lemma \ref{lem:HN-kinetic}). Our method can be applied to study the norm approximation in Fock space, for example to simplify significantly the proof in \cite{Kuz-15b}. The range $0\le \beta<1/2$  is expected to be optimal under the assumptions on the initial states in Theorem \ref{thm:main}.

The paper is organized as follows. We will revise the well-posedness of the Hartree equation \eqref{eq:Hartree-equation} and the Bogoliubov equation \eqref{eq:Bogoliubov-equation} in Section \ref{sec:well-posedness}. In section \ref{sec:Bogoliubov}, we reformulate the problem using a unitary transformation from $\gH^N$ to a truncated Fock space, following ideas in \cite{LewNamSerSol-15,LewNamSch-15}. Then we provide several estimates which are useful to implement Bogoliubov's approximation. The proof of Theorem \ref{thm:main} is presented in Section \ref{sec:main-proof}.

\section{Well-posedness of the effective equations}\label{sec:well-posedness}

From \cite[Prop. 3.3 \& Cor. 3.4]{GriMac-13} we have the following well-posedness of the Hartree equation.


\begin{lemma} \label{lem:Hartree-equation} If $u(0,\cdot)\in H^2(\R^3)$, then the Hartree equation \eqref{eq:Hartree-equation} has a unique global solution $u \in C( [0,\infty),H^2(\R^3)) \cap C^1((0,\infty),L^2(\R^3))$. Moreover, if $u(0,\cdot)\in W^{\ell,1}(\R^3)$ with $\ell$ sufficiently large, then $\|u(t,\cdot)\|_{H^2} \le C$, $\|\partial_t u(t,\cdot)\|_{L^2} \le C$ and 
$$ \|u(t, \cdot)\|_{L^\infty(\R^3)} + \|\partial_t u(t,\cdot)\|_{L^\infty(\R^3)} \le \frac{C}{(1+t)^{3/2}}$$
for a constant $C$ depending only on $\|u(0)\|_{W^{\ell,1}(\R^3)}$.
\end{lemma}  


From now on, we always assume that $u(0,\cdot)\in W^{\ell,1}(\R^3)$ with $\ell$ sufficiently large. We will also denote by $C$ a general constant depending only on $\|u(0,\cdot)\|_{W^{\ell,1}}$ (whose value can be changed from line to line). Indeed, more precisely, $C$ depends only on $\kappa_0$ in the condition $\| u(0,\cdot)\|_{W^{\ell,1}(\R^3)} \le \kappa_0$ (c.f. Theorem \ref{thm:main}).

Next, we recall the well-posedness of the Bogoliubov equation from \cite[Theorem 7]{LewNamSch-15} and \cite[Prop. 4]{NamNap-15}.


\begin{lemma} \label{lem:Bogoliubov-equation} For every initial state $\Phi(0)$ in the quadratic form domain of $\mathcal{Q}(\dGamma (1-\Delta))$, the Bogoliubov equation \eqref{eq:Bogoliubov-equation} has a unique global solution $\Phi \in C([0,\infty), \cF(\gH)) \cap L^\infty_{\rm loc} ((0,\infty), \mathcal{Q}(\dGamma (1-\Delta)))$. Moreover, if $\Phi(0)$ is a quasi-free state in $\cF_+(0)$, then $\Phi(t)$ is a quasi-free state in $\cF_+(t)$ and
$$
\langle \Phi(t), \cN \Phi(t) \rangle \le C \Big( \langle \Phi(0),\cN \Phi(0)\rangle^2 + [\log(2+t)]^2\Big).
$$
\end{lemma}

We have two remarks on the Bogoliubov equation \eqref{eq:Bogoliubov-equation}. First, although the Bogoliubov Hamiltonian $\bH(t)$ is not necessarily bounded from below and has not been defined as a self-adjoint operator, the solution $\Phi(t)$ to \eqref{eq:Bogoliubov-equation} can be still interpreted as an evolution generated by quadratic forms (see \cite[Theorems 7, 8]{LewNamSch-15} for further discussion). Second, when $\Phi(t)$ is a quasi-free state, then the Bogoliubov equation \eqref{eq:Bogoliubov-equation} becomes equivalent to the system \eqref{eq:linear-Bog-dm} (see \cite[Prop. 4]{NamNap-15} for more details), but we will not need this fact in the rest of paper.

The main new result of this section is the following kinetic estimate. 

\begin{lemma} \label{lem:bH-kinetic} Assume that $\Phi(0)$ is a quasi-free state in $\cF_+(0)$ satisfying 
$$
\big \langle \Phi(0), \dGamma(1-\Delta)  \Phi(0) \big\rangle \le  \kappa_\eps  N^{\beta+\eps} 
$$
for some $\eps>0$, where the constant $\kappa_\eps$ is independent of $N$. Then
$$
\big \langle \Phi(t), \dGamma(1-\Delta)  \Phi(t) \big\rangle \le  C_\eps  N^{\beta+\eps} ,\quad \forall t>0.
$$
\end{lemma}

Hereafter, $C_\eps$ is a general constant depending only on $\|u(0,\cdot)\|_{W^{\ell,1}}$ (more precisely, on $\kappa_0$ in the condition $\| u(0,\cdot)\|_{W^{\ell,1}} \le \kappa_0$) and $\eps$.

To prove Lemma \ref{lem:bH-kinetic}, we will need a general lower bound on the ground state energy of quadratic Hamiltonians. 

\begin{lemma} \label{lem:Bog-GSE} Let $H>0$ be a self-adjoint operator on $\gH$. Let $K:\overline{\gH}\equiv \gH^*\to \gH$ be an operator with kernel $K(x,y) \in \gH^2$. Assume that $K H^{-1} K^* \le H$ and that $H^{-1/2}K$ is Hilbert-Schmidt. Then 
$$ \dGamma(H) + \frac{1}{2} \iint \Big( K(x,y) a_x^* a_y^* + \overline{K(x,y)}a_x a_y \Big) \d x \d y \ge -\frac{1}{2} \| H^{-1/2} K\|_{\rm HS}^2.$$
\end{lemma}

This result is taken from \cite[Lemma 9]{NamNapSol-16} (see also \cite[Theorem 5.4]{BruDer-07}). Note that
$$ \| H^{-1/2} K\|_{\rm HS}^2 = \iint |H_x^{-1/2} K(x,y)|^2 \d x \d y =: \|H_x^{-1/2}K(\cdot, \cdot)\|_{L^2}^2.$$
Here we write $H_x$ to mention that the operator $H$ acts on the $x$-variable.

If we apply Lemma \ref{lem:Bog-GSE} with $H=1+\|K_2\|$ and $K=\pm K_2$, then we get
\bq \label{eq:bound-paring-dG1}
\pm \frac{1}{2} \iint \Big( K_2(t,x,y) a_x^* a_y^* + \overline{K_2(t,x,y)}a_x a_y \Big) \d x \d y \le C \cN + \frac{CN^{3\beta}}{(1+t)^3}.
\eq
Here we have used the bound on $\|K_2\|$ in \eqref{eq:norm-K2} and 
\begin{align} \label{eq:L2-K2-N3beta}
\|K_2(t,\cdot,\cdot)\|_{L^2}^2 &\le \|\widetilde K_2(t,\cdot,\cdot)\|_{L^2}^2  = \iint | u(t,x)|^2 |w_N(x-y)|^2 |u(t,y)|^2 \d x \d y \nn\\
&\le \|w_N\|_{L^2}^2 \|u(t,\cdot)\|_{L^\infty}^2  \|  u(t,\cdot)\|_{L^2}^2 \le \frac{CN^{3\beta}}{(1+t)^3}
\end{align}
by Lemma \ref{lem:Hartree-equation}. In order to improve the factor $N^{3\beta}$ in \eqref{eq:bound-paring-dG1}, we will apply Lemma \ref{lem:Bog-GSE} with $H=1-\Delta$. We will need the following estimate.

\begin{lemma}\label{lem:Sobolev-inverse-K2} For all $\eps>0$ we have 
\begin{align*}
\|(1-\Delta_x)^{-1/2} K_2(t, \cdot, \cdot)\|^2_{L^2} + \|(1-\Delta_x)^{-1/2} \partial_t K_2(t, \cdot, \cdot)\|^2_{L^2} \le \frac{C_\eps N^{\beta+\eps}}{(1+t)^3}.
\end{align*}
\end{lemma}

\begin{proof} We will present a detailed proof for $\partial_t K_2(t)$ and $K_2(t)$ can be treated by the same way. Recall that 
$$ K_2(t,\cdot ,\cdot)=Q(t)\otimes Q(t) \widetilde K_2(t,\cdot,\cdot), \quad \widetilde K_2(t,x,y)=u(t,x)w_N(x-y) u(t,y).$$
Hence,  
$$ \partial_t K_2(t)= \partial_t Q(t) \otimes Q(t) \widetilde K_2(t) + Q(t) \otimes \partial_t Q(t) \widetilde K_2(t)+Q(t)\otimes Q(t) \partial_t \widetilde K_2(t).$$
Since $\partial_t Q(t)=-| \partial_t u(t)\rangle \langle u(t)| -|u(t)\rangle \langle \partial_t u(t)|$, we have
\begin{align*}
&\|\partial_t Q(t) \otimes Q(t) \widetilde K_2(t,\cdot,\cdot)\|_{L^2} \le \| (\partial_t Q(t) \otimes 1) \widetilde K_2(t,\cdot,\cdot)\|_{L^2} \\
& \le \| (| \partial_t u(t)\rangle \langle u(t)| \otimes 1 )  \widetilde K_2(t,\cdot,\cdot)\|_{L^2} + \| (|  u(t)\rangle \langle \partial_t u(t)|  \otimes 1 ) \widetilde K_2(t,\cdot,\cdot)\|_{L^2}
\end{align*}
Using Lemma \ref{lem:Hartree-equation} and $\|w_N\|_{L^1}=\|w\|_{L^1}$, it is straightforward to see that
\begin{align*} 
&\left\| (|\partial_t u\rangle \langle u|  \otimes 1) \widetilde K_2(t, \cdot, \cdot) \right\|_{L^2}^2 \nn\\
&=\iint  \left| \int \overline{u(t,z)}  u(t,z) w_N(z-y) u(t,y) \d z \right|^2 |\partial_t u(t,x)|^2 \d x \d y \nn\\
& \le \|u(t,\cdot)\|_{L^\infty}^4 \|w_N\|_{L^1}^2 \|u(t,\cdot)\|_{L^2}^2   \|\partial_t u(t,\cdot)\|_{L^2}^2  \le \frac{C}{(1+t)^{3}}.
\end{align*}
combining this with similar estimates, we find that
\begin{align} \label{eq:L2-dtQ-K2}
&\|\partial_t Q(t) \otimes Q(t) \widetilde K_2(t,\cdot,\cdot)\|_{L^2} + \| Q(t) \otimes \partial_t Q(t) \widetilde K_2(t,\cdot,\cdot)\|_{L^2} \nn \\
& \le \|(\partial_t Q(t) \otimes 1) \widetilde K_2(t,\cdot,\cdot)\|_{L^2} + \| (1 \otimes \partial_t Q(t)) \widetilde K_2(t,\cdot,\cdot)\|_{L^2}  \le \frac{C}{(1+t)^{3/2}}.
\end{align}
By the same argument, we also obtain
$$\| (1-Q(t)\otimes Q(t)) \partial_t \widetilde K_2(t,\cdot,\cdot)\|_{L^2} \le \frac{C}{(1+t)^{3/2}}.$$

Note that $(1-\Delta_x)^{-1/2} \le 1$ on $L^2$, and hence we can insert $(1-\Delta_x)^{-1/2}$ into the above $L^2$ norm estimates for free. It remains to show that 
\bq \label{eq:eq:dt-K2-half}
\|(1-\Delta_x)^{-1/2} \partial_t \widetilde K_2(t, \cdot, \cdot)\|^2_{L^2} \le \frac{C_\eps N^{\beta+\eps}}{(1+t)^3}, \quad \forall \eps>0.
\eq
Similarly to \eqref{eq:L2-K2-N3beta}, we have
$$
\|\partial_t \widetilde K_2(t,\cdot,\cdot) \|_{L^2}^2 \le  	\frac{CN^{3\beta}}{(1+t)^3}.
$$
Therefore, by interpolation (more precisely, by H\"older's inequality in Fourier space), \eqref{eq:eq:dt-K2-half} follows from the following estimate
\begin{align} \label{eq:1-Delta-K2-075}
\|(1-\Delta_x)^{-3/4-\eps} \partial_t \widetilde K_2(t, \cdot, \cdot)\|^2_{L^2}  \le \frac{C_\eps}{(1+t)^3}, \quad \forall \eps>0.
\end{align}
It suffices to show that
\begin{align} \label{eq:1-Delta-075}
\|(1-\Delta_x)^{-3/4-\eps} f(t, \cdot, \cdot)\|^2_{L^2}  \le \frac{C_\eps}{(1+t)^3}, \quad \forall \eps>0
\end{align}
with $f(t,x,y)=\partial_t u(t,x) w_N(x-y) u(t,y)$. The bound \eqref{eq:1-Delta-075} can be proved using an argument in \cite{GriMac-13}. Let us compute the Fourier transform:
\begin{align*}
\widehat{f}(t,p,q) &= \iint u(t,x)w_N(x-y)(\partial_t u)(t,y) e^{-2\pi i (p\cdot x + q\cdot y)} \d x \d y \\
&=\iint u(t,y+z) w_N(z) (\partial_t u)(t,y) e^{-2\pi i (p\cdot (y+z) + q\cdot y)} \d z\d y \\
&=\int w_N(z) \widehat{(u_z \partial_t u)}(t,p+q) e^{-2\pi ip\cdot z}\d z
\end{align*}
where $u_z(t,\cdot):=u(t,z+\cdot).$ By the Cauchy-Schwarz inequality,
\begin{align*}
\left|\widehat{f}(t,p,q) \right|^2 \le \|w_N\|_{L^1} \int |w_N(z)| \cdot| \widehat{(u_z \partial_t u)}(t,p+q)|^2 \d z.
\end{align*}
Using Plancherel's Theorem, we can estimate
\begin{align*}
& \|(1-\Delta_x)^{-3/4-\eps} f(t, \cdot, \cdot)\|^2_{L^2} = \iint (1+|2\pi p|^2)^{-3/2-2\eps}\left|\widehat{f}(t,p,q) \right|^2 \d p \d q \\
&\le \|w_N\|_{L^1}  \iiint (1+|2\pi p|^2)^{-3/2-2\eps} |w_N(z)| \cdot |\widehat{(u_z \partial_t u)}(t,p+q)|^2 \d p \d q \d z.
\end{align*}
By Lemma \ref{lem:Hartree-equation}, 
\begin{align*}
\int |\widehat{(u_z \partial_t u)}(t,p+q)|^2 \d q & = \| (u_z \partial_t u)(t,\cdot) \|_{L^2}^2 \\
& \le \|u(t,\cdot)\|_{L^\infty}^2 \|\partial_t u(t,\cdot)\|_{L^2}^2 \le \frac{C}{(1+t)^3}.
\end{align*}
Therefore,  \eqref{eq:1-Delta-075} follows from $\|w_N\|_{L^1}=\|w\|_{L^1}$ and the fact that
\begin{equation*}
\int  (1+|2\pi p|^2)^{-3/2-2\eps}  \d p \le C_\eps <\infty.
\end{equation*}
Thus \eqref{eq:1-Delta-075} holds true. By the same argument, we obtain a similar inequality with $f(t,x,y)$ replaced by $u(t,x)w_N(x-y) \partial_t u(t,y)$. Combining these two estimates, we deduce \eqref{eq:1-Delta-K2-075}. This completes the proof.
\end{proof}

Now we apply Lemmas \ref{lem:Bog-GSE} and \ref{lem:Sobolev-inverse-K2} to bound $\bH(t)$.
 
\begin{lemma} \label{lem:bHt-dbHt} For every $\eps>0$ and $\eta>0$, we have
\begin{align*}
\pm \Big( \bH(t) + \dGamma(\Delta) \Big) &\le \eta \dGamma(1-\Delta) +  \frac{C_\eps(\cN+N^{\beta+\eps})}{\eta(1+t)^3},\\
\pm \partial_t \bH(t) &\le \eta \dGamma(1-\Delta) +  \frac{C_\eps (\cN  +  N^{\beta+\eps})}{ \eta (1+t)^3},\\
\pm i[\bH(t),\cN] &\le \eta \dGamma(1-\Delta) + \frac{C_\eps (\cN  +  N^{\beta+\eps})}{ \eta (1+t)^3},
\end{align*}
as quadratic forms on $\cF(\gH)$. The constant $C_\eps$ is independent of $\eta$.
\end{lemma}

\begin{proof} First, we consider  
$$
\bH(t) + \dGamma(\Delta)= \dGamma(h+\Delta) + \frac{1}{2} \iint \Big( K_2(t,x,y) a_x^* a_y^* + \overline{K_2(t,x,y)}a_x a_y \Big) \d x \d y.
$$
Recall that 
$$h+\Delta = |u(t,\cdot)|^2\ast w_N -\mu_N(t) + Q(t) \widetilde K_1(t) Q(t),$$
where $\widetilde K_1(t)$ is the operator with kernel $\widetilde K_1(t)=u(t,x)w_N(x-y)\overline{u(t,y)}$. Using Lemma \ref{lem:Hartree-equation}, we have
\begin{align} \label{eq:K1-norm-1}
&\| |u(t,\cdot)|^2*w_N \|_{L^\infty} \le \|u(t,\cdot)\|_{L^\infty}^2 \|w_N\|_{L^1} \le \frac{C}{(1+t)^3},\\
&\mu_N(t) = \frac{1}{2} \int |u(t,x)|^2 (|u(t,\cdot)|^2*w_N)(x) \d x  \le \frac{C}{(1+t)^3}.\label{eq:K1-norm-2}
\end{align}
Moreover, $\| Q(t) \widetilde K_1(t) Q(t) \| \le \| \widetilde K_1(t) \|$ and  
\begin{align} \label{eq:K1-norm}
&\| \widetilde K_1(t) \| = \sup_{ \|f\|_{L^2}=1} \left| \iint \overline{f(x)} u(t,x) w_N(x-y) \overline{u(t,y)}f(y) \d x \d y \right| \\
& \le \sup_{\|f\|_{L^2}=1} \|u(t,\cdot)\|_{L^\infty}^2 \iint \frac{|f(x)|^2+|f(y)|^2}{2} w_N(x-y) \d x \d y  \le \frac{C}{(1+t)^3}.\nn
\end{align}
From \eqref{eq:K1-norm-1}, \eqref{eq:K1-norm-2} and \eqref{eq:K1-norm} and the triangle inequality, we get 
\bq \label{eq:norm-h}
\| h + \Delta \|  \le  \frac{C}{(1+t)^3}.
\eq
Similarly to \eqref{eq:K1-norm}, we can bound the operator $K_2(t):\gH^*\to \gH$ as 
\bq \label{eq:norm-K2}
\|K_2(t)\| \le \frac{C}{(1+t)^3}.
\eq
Now we apply Lemma \ref{lem:Bog-GSE} with $H=\eta(1-\Delta)+\eta^{-1}\|K_2\|^2$ and $K=\pm K_2$, where $\eta>0$ is arbitrary. Since $H \ge \|K\|$, we have
$$H \ge \|K\| \ge K \|K\|^{-1}K^* \ge KH^{-1}K .$$
Moreover, using $H \ge \eta(1-\Delta)$ and Lemma \ref{lem:Sobolev-inverse-K2}, we get
$$
\|H^{-1/2}K\|_{\rm HS}^2 \le \eta^{-1}\|(1-\Delta)^{-1/2}K_2\|_{\rm HS}^2 \le \frac{C_\eps N^{\beta+\eps}}{\eta(1+t)^3}.
$$
Therefore, Lemma \ref{lem:Bog-GSE} implies that 
\begin{align} \label{eq:pairing-dGD}
&\pm \frac{1}{2} \iint \Big( K_2(t,x,y) a_x^* a_y^* + \overline{K_2(t,x,y)}a_x a_y \Big) \d x \d y \nn \\
& \le \eta\dGamma(1-\Delta) + \eta^{-1}\|K_2\|^2 \cN + \frac{C_\eps N^{\beta+\eps}}{\eta(1+t)^3}.
\end{align}
From \eqref{eq:norm-h}, by using the Cauchy-Schwarz inequality $C(1+t)^{-3} \le \eta + C\eta^{-1}(1+t)^{-6}$ we get
$$
\pm \dGamma(h+\Delta) \le \eta \cN + \frac{C \cN}{\eta(1+t)^{6}}
$$
for all $\eta>0$. Combining this with \eqref{eq:pairing-dGD} (and the obvious bound $\cN\le \dGamma(1-\Delta)$) we conclude that
$$
\pm \Big( \bH(t) + \dGamma(\Delta) \Big) \le \eta \dGamma(1-\Delta) +  \frac{C_\eps (\cN+N^{\beta+\eps})}{\eta(1+t)^3}.
$$

The bound on $\partial_t \bH(t)$ is obtained by the same way. Indeed, by Lemma \ref{lem:Hartree-equation},
\begin{align} \label{eq:der-K1-1}
& \| \partial_t (|u(t,\cdot)|^2\ast w_N)  \|_{L^\infty} \le \| \partial_t  |u(t,\cdot)|^2 \|_{L^\infty} \|w_N\|_{L^1} \le \frac{C}{(1+t)^3},\\
&| \partial_t \mu_N(t) | \le 2 \iint |\partial_t (|u(t,x)|^2)| w_N(x-y) |u(t,y)|^2 \d x \d y \le \frac{C}{(1+t)^3}, \label{eq:der-K1-2}\\
\label{eq:der-K1}&\| \partial_t (Q(t) \widetilde K_1(t) Q(t) ) \| \le  \|\partial_t \widetilde K_1(t)\| + 2 \|\partial_t Q(t)\| \cdot \|\widetilde K_1(t)\|  \le \frac{C}{(1+t)^3}.
\end{align}
Similarly to \eqref{eq:der-K1}, we can bound the operator $\partial_t K_2(t): \gH^*\to \gH$ as 
\bq \label{eq:norm-dt-K2}
\| \partial_t K_2(t)\| \le \frac{C}{(1+t)^3}.
\eq
Then we apply Lemma \ref{lem:Bog-GSE} with $H=\eta(1-\Delta) + \eta^{-1}\|\partial_t K_2\|^2$ and $K= \pm \partial_t K_2$ and obtain 
\begin{equation*}
\pm \partial_t \bH(t) \le \eta \dGamma(1-\Delta) +  \frac{C_\eps( \cN  + N^{\beta+\eps})}{ \eta (1+t)^3}. 
\end{equation*}

Finally, since $[a_x^*a^*_y,\cN]= -2 a_x^* a_y^*$, we have
\begin{align*}
i[\bH(t),\cN] &= \frac{i}{2} \iint \Big( K_2(t,x,y) [a_x^* a_y^*,\cN]  + \overline{K_2(t,x,y)}[a_x a_y,\cN] \Big) \d x \d y \\
&= - \iint \Big( iK_2(t,x,y) a_x^* a_y^* + \overline{iK_2(t,x,y)}a_x a_y \Big) \d x \d y
\end{align*}
Using Lemma \ref{lem:Bog-GSE} again, we obtain \eqref{eq:pairing-dGD} with $K_2$ replaced by $-2iK_2$, namely
\begin{equation*}
\pm i[\bH(t),\cN] \le \eta \dGamma(1-\Delta) + \frac{C_\eps (\cN + N^{\beta +\eps})}{\eta(1+t)^3}
\end{equation*}
for all $\eta>0$. This completes the proof.
\end{proof}

\begin{remark} \label{rmk:ground-state} Assume that $\Phi$ is the ground state of the quadratic Hamiltonian
$$ \bH_V(0) = \dGamma(h(0)+V) + \frac{1}{2} \iint \Big( K_2(0,x,y) a_x^* a_y^* + \overline{K_2(0,x,y)}a_x a_y \Big) \d x \d y,$$
where $V(x)$ is an appropriate trapping potential which ensures that $h(0)+V \ge \eta >0$ (in particular, this implies that $h(0)+V \ge C_\eta^{-1} (1-\Delta)$ for some constant $C_\eta>0$). By Lemma \ref{lem:Bog-GSE}, we know that 
\begin{align*}
\bH_V(0) &\ge \frac{1}{2}\dGamma(h(0)+V) + \frac{1}{2} \iint \Big( K_2(0,x,y) a_x^* a_y^* + \overline{K_2(0,x,y)}a_x a_y \Big) \d x \d y \\
&\ge - C \|(1-\Delta_x)^{-1/2} K_2(0,\cdot,\cdot)\|_{L^2}^2  \ge -C_\eps N^{\beta+\eps}.
\end{align*}
In particular, $\bH_V(0)$ is bounded from below and it can be diagonalized by a Bogoliubov transformation. Moreover, $\Phi$ is a quasi-free state 
and 
$$ \langle \Phi, \cN \Phi \rangle \le C  \|(1-\Delta_x)^{-1} K_2(0,\cdot,\cdot)\|_{L^2}^2 \le C.$$
Here the first bound is a consequence of \cite[Theorem 1 (ii)]{NamNapSol-16}.

Since the ground state energy of $\bH_V(0)$ is always non-positive (see \cite[Theorem 2.1 (i)]{LewNamSerSol-15}), we have 
$$
0 \ge \langle \Phi, \bH_V(0) \Phi\rangle \ge \frac{1}{2}\langle \Phi, \dGamma(h(0)+V) \Phi \rangle - C_\eps N^{\beta+\eps}.
$$
Combining with $h(0)+V \ge C_\eta^{-1} (1-\Delta)$, we obtain
$$\langle \Phi, \dGamma(1-\Delta) \Phi \rangle \le C_\eps N^{\beta+\eps}, \quad \forall \eps>0.$$ 
This motivates the assumptions on $\Phi(0)$ in Theorem \ref{thm:main}.
\end{remark}

We are ready to give 

\begin{proof}[Proof of Lemma \ref{lem:bH-kinetic}] From the Bogoliubov equation \eqref{eq:Bogoliubov-equation}, we have
$$
\partial_t \big\langle \Phi(t), \bH (t) \Phi(t)  \big\rangle = \big\langle \Phi(t), \partial_t \bH (t)  \Phi(t)  \big\rangle
$$
which implies that 
\begin{align} \label{eq:Phit-Grw}
\big\langle \Phi(t), \bH (t) \Phi(t)  \big\rangle - \big\langle \Phi(0), \bH (0) \Phi(0)  \big\rangle = \int_0^t \big\langle \Phi(s), \partial_s \bH (s) \Phi(s)  \big\rangle \d s. 
\end{align}
From the first bound in Lemma \ref{lem:bHt-dbHt} with $\eta=1/2$, we get
\begin{align*} 
\pm \big\langle \Phi(t), (\bH (t) + \dGamma(\Delta) ) \Phi(t)  \big\rangle  &\le  \frac{1}{2} \big \langle \Phi(t), \dGamma(1-\Delta)  \Phi(t) \big\rangle \\
&\qquad \qquad  + C_\eps \Big( \big \langle \Phi(t), \cN  \Phi(t) \big\rangle +  N^{\beta+\eps} \Big). 
\end{align*}
This implies 
\begin{align} \label{eq:kin-Phi-1}
\big\langle \Phi(0), \bH (0) \Phi(0)  \big\rangle &\le \frac{3}{2}\big \langle \Phi(0), \dGamma(1-\Delta)  \Phi(0) \big\rangle + \nn\\
&\qquad  + C_\eps \Big( \big \langle \Phi(0), \cN  \Phi(0) \big\rangle + N^{\beta+\eps} \Big) \le C_\eps N^{\beta+\eps}
\end{align}
(where we have used the assumption on $\Phi(0)$) and 
\begin{align} 
\big\langle \Phi(t), \bH (t) \Phi(t)  \big\rangle  \ge  \frac{1}{2} \big \langle \Phi(t), \dGamma(1-\Delta)  \Phi(t) \big\rangle - C_\eps \Big( \big \langle \Phi(t), \cN  \Phi(t) \big\rangle +  N^{\beta+\eps} \Big). \label{eq:kin-Phi-2}
\end{align}
Next, from  the second bound in Lemma \ref{lem:bHt-dbHt} with $\eta=(1+t)^{-3/2}$ we have    
\begin{align}
\big\langle \Phi(t), \partial_t \bH (t)\Phi(t)  \big\rangle \le C_\eps  \frac{ \big\langle \Phi(t), \dGamma(1-\Delta) \Phi(t) \big\rangle+ N^{\beta+\eps}}{(1+t)^{3/2}} . \label{eq:kin-Phi-3}
\end{align}
Inserting \eqref{eq:kin-Phi-1}, \eqref{eq:kin-Phi-2} and \eqref{eq:kin-Phi-3} into \eqref{eq:Phit-Grw}, we obtain
\begin{align} \label{eq:Phit-Grw-1}
\big \langle \Phi(t), \dGamma(1-\Delta)  \Phi(t) \big\rangle & \le  C_\eps \int_0^t \frac{\big \langle \Phi(s), \dGamma(1-\Delta)  \Phi(s) \big\rangle}{(1+s)^{3/2}} \d s \nn\\ &\qquad \qquad + C_\eps \Big( \big \langle \Phi(t), \cN  \Phi(t) \big\rangle + N^{\beta+\eps} \Big). 
\end{align}

Now instead of using the bound on $\big \langle \Phi(t), \cN  \Phi(t) \big\rangle$ in Lemma \ref{lem:Bogoliubov-equation}, we present another argument which will be used again to deal with the many-body Schr\"odinger evolution in Section \ref{sec:Bogoliubov}.  From the Bogoliubov equation \eqref{eq:Bogoliubov-equation} and the third bound in Lemma \ref{lem:bHt-dbHt} with $\eta=(1+t)^{-3/2}$, it follows that
\begin{align*}
\partial_t \langle \Phi(t), \cN \Phi(t) \rangle = \langle \Phi(t), i[\bH(t),\cN] \Phi(t) \rangle \le C_\eps  \frac{ \big\langle \Phi(t), \dGamma(1-\Delta) \Phi(t) \big\rangle+ N^{\beta+\eps}}{(1+t)^{3/2}}.
\end{align*}
Integrating over $t$ and using the assumption on $\Phi(0)$ we have
\begin{align*}
\langle \Phi(t), \cN \Phi(t) \rangle \le C_\eps \int_0^t \frac{ \big\langle \Phi(s), \dGamma(1-\Delta) \Phi(s) \big\rangle}{(1+s)^{3/2}} + C_\eps N^{\beta+\eps}.
\end{align*}
Inserting the latter inequality into the right side of \eqref{eq:Phit-Grw-1} we obtain
\begin{align} \label{eq:Phit-Grw-2}
\big \langle \Phi(t), \dGamma(1-\Delta)  \Phi(t) \big\rangle & \le C_\eps \int_0^t \frac{\big \langle \Phi(s), \dGamma(1-\Delta)  \Phi(s) \big\rangle}{(1+s)^{3/2}} \d s + C_\eps N^{\beta+\eps}. 
\end{align}
Now we define 
$$ f(t) := \int_0^t \frac{\big\langle \Phi(s), \dGamma(1-\Delta) \Phi(s)  \big\rangle}{(1+s)^{3/2}} \d s + N^{\beta+\eps}$$
and rewrite \eqref{eq:Phit-Grw-2} as   
$$
\frac{\d}{\d t} \log (f(t)) = \frac{f'(t)}{f(t)} \le  \frac{C_\eps}{(1+t)^{3/2} }.
$$
Integrating over $t$ and using the fact that $(1+t)^{-3/2}$ is integrable on $(0,\infty)$, we get  $f(t)\le C_\eps N^{\beta+\eps}$. The desired result then follows from \eqref{eq:Phit-Grw-2}.
\end{proof}


\section{Bogoliubov approximation} \label{sec:Bogoliubov}

Recall that any vector $\Psi\in \gH^N$ can be decomposed uniquely as 
\begin{equation*}
\Psi=\sum_{n=0}^N u(t)^{\otimes (N-n)} \otimes_s \psi_n  = \sum_{n=0}^N \frac{(a^*(u(t)))^{N-n}}{\sqrt{(N-n)!}} \psi_n
\end{equation*}
with $\psi_n \in \gH_+(t)^{n}$, see ~\cite[Sec. 2.3]{LewNamSerSol-15}. This gives rise the  unitary operator
\begin{equation*}
\begin{array}{cccl}
U_{N}(t): & \gH^N & \to & \displaystyle \cF_+^{\le N}(t):=\bigoplus_{n=0}^N \gH_+(t)^n \\[0.3cm]
 & \Psi & \mapsto & \psi_0\oplus \psi_1 \oplus\cdots \oplus \psi_N.
\end{array}
\end{equation*}
Following \cite{LewNamSch-15}, we reformulate the Schr\"odinger equation $\Psi_N(t)=e^{-itH_N} \Psi(0)$ by introducing 
$$\Phi_N(t):=U_N(t) \Psi_N(t)$$
which belongs to $\cF_+^{\le N}(t)$ and satisfies the equation
\bq \label{eq:eq-PhiNt}
\left\{
\begin{aligned}
i \partial_t \Phi_N(t)  &=  \widetilde H_N (t)   \Phi_N(t), \\
 \Phi_N(0) & = \1^{\le N} \Phi(0).
\end{aligned}
\right.
\eq
Here $\1^{\le N}$ is the projection onto $\cF^{\le N}=\mathbb{C} \oplus \gH \oplus \cdots \oplus \gH^N$ and 
$$
\widetilde H_N (t)=  \1^{\le N} \Big[ \bH(t) + \frac{1}{2}\sum_{j=0}^4 ( R_{j} + R_j^*) \Big] \1^{\le N}
$$
with 
\begin{align*}
R_{0}&=R_0^*= \d\Gamma(Q(t)[w_N*|u(t)|^2+ \widetilde{K}_1(t) -\mu_N(t)]Q(t))\frac{1-\cN}{N-1},\\
R_{1}&=-2\frac{\cN\sqrt{N-\cN}}{N-1} a(Q(t)[w_N*|u(t)|^2]u(t)),\\
R_{2}&= \iint  K_2(t,x,y) a^*_x a^*_y \d x \d y \left(\frac{\sqrt{(N-\cN)(N-\cN-1)}}{N-1}-1\right),\\
R_{3}& =\frac{\sqrt{N-\cN}}{N-1}\iiiint( 1 \otimes Q(t) w_N Q(t)\otimes Q(t))(x,y;x',y')\times \\
& \qquad \qquad \qquad \qquad \qquad \qquad \qquad  \times \overline{u(t,x)} a^*_y a_{x'} a_{y'} \,\d x \d y \d x' \d y',\\
R_{4}&=R_4^*= \frac{1}{2(N-1)}\iiiint({Q(t)}\otimes{Q(t)}w_N Q(t)\otimes Q(t))(x,y;x',y')\times \\
&  \qquad \qquad \qquad \qquad \qquad \qquad \qquad \qquad \qquad   \times a^*_x a^*_y a_{x'} a_{y'} \,\d x \d y \d x' \d y'.
\end{align*}
Here, in $R_0$ and $R_1$ we write $w_N$ for the function $w_N(x)$, while in $R_3$ and $R_4$ we write $w_N$ for the two-body multiplication operator $w_N(x-y)$.

In order to compare $\Phi_N(t)$ with the Bogoliubov dynamics $\Phi(t)$, we need to bound all error terms $R_j$'s. In \cite[Prop. 3]{NamNap-15}, we proved that 
$$ (R_j+R_j^*) \1^{\le N} (R_j+R_j^*) \le C ( N^{6\beta-2} + N^{3\beta-1}) (1+\cN)^4.$$
Unfortunately, this bound is only useful when $\beta<1/3$. In the present paper, we will derive several improved estimates. Let us start with

\begin{lemma}\label{lem:Rj} We have the quadratic form estimates on $\cF_+^{\le N}$: 
\begin{align*}
\pm (R_j+R_j^*) \le \eta \Big( R_4 + \frac{\cN^2}{N} \Big)+ \frac{C (1+\cN)}{\eta(1+t)^3}, \quad \forall \eta>0, \,\forall j=0,1,2,3,
\end{align*}
and
$$ 0\le R_4 \le CN^{3\beta-1} \cN^2, \quad R_4\le CN^{\beta-1} \dGamma(-\Delta) \cN.$$
Here the constant $C$ depends only on $\|u(0,\cdot)\|_{W^{1,\ell}}$ (more precisely, on $\kappa_0$ in the condition $\| u(0,\cdot)\|_{W^{\ell,1}} \le \kappa_0$).  
\end{lemma}

\begin{proof} Let us go term by term.

\noindent
$\boxed{j=0}$ Recall that
$$
R_0 = \d\Gamma\Big( Q(t)[w_N*|u(t)|^2+\widetilde{K}_1(t) -\mu_N(t)]Q(t) \Big)\frac{1-\cN}{N-1}.
$$
From the operator bounds in \eqref{eq:K1-norm-1}-\eqref{eq:K1-norm-2}-\eqref{eq:K1-norm}, we have
$$
\pm \d\Gamma\Big( Q(t)[w_N*|u(t)|^2+\widetilde{K}_1(t) -\mu_N(t)]Q(t) \Big) \le \frac{C \cN}{(1+t)^3}.
$$
Since the left side of the latter inequality commutes with $\cN$,  we get
\begin{align} \label{eq:R0-final}
\pm R_0 \le \frac{C\cN^2}{N(1+t)^3} &\le \eta \frac{\cN^2}{N} + \eta^{-1} \frac{C\cN^2}{N(1+t)^6} \nn\\
&\le  \eta \frac{\cN^2}{N} + \eta^{-1} \frac{C\cN}{(1+t)^6}, \quad \forall \eta>0.
\end{align}
In the last inequality, we have used the fact that $\cN \le N$ on $\cF_{+}^{\le N}(t)$.

\medskip

\noindent 
$\boxed{j=1}$ For every $\Phi\in \cF_+^{\le N}(t)$, by the Cauchy-Schwarz inequality we have
\begin{align*}
\left| \langle \Phi,  R_{1}\Phi \rangle \right| &= \frac{2}{N-1} \left| \left\langle \Phi, \cN\sqrt{N-\cN} a\Big( Q(t)[w_N*|u(t)|^2]u(t)  \Big) \Phi \right\rangle \right| \\
& \le \frac{2}{N-1} \| \cN\sqrt{N-\cN} \Phi\| \Big\| a \Big( Q(t)[w_N*|u(t)|^2]u(t) \Big) \Phi \Big\|.
\end{align*}
Now we use the elementary inequality $a^*(v)a(v)\le \|v\|_{L^2}^2 \cN$ and 
\begin{align*}
\Big\|  Q(t)[w_N*|u(t)|^2]u(t) \Big\|_{L^2}  \le \| [w_N*|u(t)|^2] u(t)\|_{L^2} \le \frac{C}{(1+t)^{3}}.
\end{align*}
Here the last estimate is \eqref{eq:K1-norm-1}. Thus
\begin{align}\label{eq:R1-final-0}
\left| \langle \Phi,  R_{1} \Phi \rangle \right| &\le \frac{C}{N(1+t)^{3/2}} \langle \Phi, \cN^2(N-\cN)  \Phi\rangle^{1/2} \langle \Phi, \cN \Phi \rangle^{1/2} \nn\\
&\le  \frac{\eta}{N} \langle \Phi, \cN^2 \Phi \rangle + \frac{C}{\eta(1+t)^3}  \langle \Phi, \cN \Phi \rangle, \quad \forall \eta>0.
\end{align}
In the last estimate we have used $0\le N-\cN \le N$ on $\cF_+^{\le N}(t)$. Consequently,
\begin{align*}
\pm \langle \Phi,  (R_{1}+R_1^*) \Phi \rangle = \pm 2 \Re \langle \Phi, R_1 \Phi \rangle \le  \frac{\eta}{N} \langle \Phi, \cN^2 \Phi \rangle + \frac{C}{\eta(1+t)^3}  \langle \Phi, \cN \Phi \rangle
\end{align*}
for all $\Phi\in \cF_+^{\le N}(t)$. Therefore, 
\bq \label{eq:R1-final}
\pm (R_1+R_1^*) \le \eta \frac{\cN^2}{N} + \eta^{-1} \frac{C\cN}{(1+t)^3}, \quad \forall \eta>0.
\eq

\noindent
$\boxed{j=2}$ For every $\Phi \in \cF_+^{\le N}(t)$, we have 
\begin{align*}
\langle \Phi, R_2 \Phi \rangle &= \iint  K_2(t,x,y) \Big\langle \Phi, a^*_x a^*_y  \Big(\frac{\sqrt{(N-\cN)(N-\cN-1)}}{N-1}-1\Big) \Phi \Big\rangle \d x \d y \\
&= \iint \widetilde K_2(t,x,y) \Big\langle \Phi, a^*_x a^*_y  \Big(\frac{\sqrt{(N-\cN)(N-\cN-1)}}{N-1}-1\Big) \Phi \Big\rangle \d x \d y.
\end{align*}
Here we can replace $K_2(t)=Q(t)\otimes Q(t) \widetilde K_2(t)$ by $\widetilde K_2(t)$, namely replace the projection $Q(t)$ by the identity, because $\Phi$ belongs to the excited Fock space $\cF_+(t)$ (putting differently, this is because $a(u)\Phi=0$). By the Cauchy-Schwarz inequality, we can estimate 
\begin{align*} 
 \left| \langle \Phi, R_2 \Phi \rangle \right| & \le \iint |u(t,x)| w_N(x-y) |u(t,y)| \| a_x a_y \Phi \|  \times \nn\\
&\qquad \qquad \qquad \times  \left\| \Big(\frac{\sqrt{(N-\cN)(N-\cN-1)}}{N-1}-1\Big) \Phi  \right\| \d x \d y.
\end{align*}
Using $0\le \cN \le N$ on $\Phi\in \cF_+^{\le N}(t)$, it is straightforward to see that
\bq \label{eq:sqrt-N-N-1}
\left\| \Big(\frac{\sqrt{(N-\cN)(N-\cN-1)}}{N-1}-1\Big) \Phi \right\|   \le \frac{C}{\sqrt{N}} \langle \Phi, (1+\cN) \Phi\rangle^{1/2}.
\eq
Moreover, by the Cauchy-Schwarz inequality again,
\begin{align*} 
&\iint |u(t,x)| w_N(x-y) |u(t,y)| \| a_x a_y \Phi \|  \d x \d y \nn\\
&\le \left( \iint |u(t,x)|^2 w_N(x-y) |u(t,y)|^2 \d x \d y \right)^{1/2} \nn \\
&\qquad  \times  \left( \iint w_N(x-y) \| a_x a_y \Phi \|^2 \d x \d y \right)^{1/2}   \le \frac{C\sqrt{N}}{(1+t)^{3/2}} \langle \Phi, R_4 \Phi\rangle^{1/2}.
\end{align*}
In the last estimate we have used \eqref{eq:K1-norm-2} and the definition of $R_4$. Thus
\begin{align}
\left| \langle \Phi, R_2\Phi \rangle \right| &\le \frac{C}{(1+t)^{3/2}} \langle \Phi, R_4 \Phi\rangle^{1/2} \langle \Phi, (1+\cN) \Phi\rangle^{1/2} \nn \\
& \le \eta \langle \Phi, R_4 \Phi\rangle + \frac{C}{\eta (1+t)^3} \langle \Phi, (1+\cN) \Phi\rangle, \quad \forall \eta>0. \label{eq:R2-final-0}
\end{align}
Consequently,
\bq \label{eq:R2-final}
\pm (R_2+R_2^*)  \le \eta R_4 + \frac{C (1+\cN)}{\eta (1+t)^3}  , \quad \forall \eta>0.
\eq

\medskip

\noindent
$\boxed{j=3}$ For all $\Phi\in \cF_+^{\le N}(t)$, by using the simplification involving the projection $Q(t)$ as above and the Cauchy-Schwarz inequality we have 
\begin{align} \label{eq:R3-final-0}  
\left| \langle \Phi, R_3 \Phi \rangle \right| &=  \frac{1}{N-1}  \left| \iint  w_N(x-y) \overline{u(t,x)} \Big \langle \Phi,\sqrt{N-\cN} a^*_y a_{y} a_x \Phi \Big\rangle \,\d x \d y \right| \nn\\
& \le \frac{1}{N-1} \iint  w_N(x-y) |u(t,x)| \cdot \| a_y \sqrt{N-\cN} \Phi \| \cdot \| a_{y}a_x  \Phi \| \d x \d y \nn\\
& \le \frac{2\|u(t,\cdot)\|_{L^\infty}}{N-1} \left( \iint w_N(x-y) \|a_x a_y \Phi\|^2 \d x \d y\right)^{1/2} \nn \\
&\qquad\qquad\qquad \times \left( \iint w_N(x-y) \| a_y \sqrt{N-\cN} \Phi \|^2 \d x \d y \right)^{1/2} \nn\\
&\le  \frac{C}{(1+t)^{3/2}} \langle \Phi, R_4 \Phi \rangle^{1/2} \langle \Phi, \cN \Phi \rangle^{1/2} \nn\\
&\le \eta \langle \Phi, R_4 \Phi \rangle + \frac{C}{\eta (1+t)^3} \langle \Phi, \cN \Phi \rangle, \quad \forall \eta>0.
\end{align}
Thus we conclude that  
\bq \label{eq:R3-final}
\pm (R_3+R_3^*)  \le \eta R_4 + \frac{C\cN}{\eta(1+t)^3}   , \quad \forall \eta>0.
\eq

Collecting  \eqref{eq:R0-final}, \eqref{eq:R1-final}, \eqref{eq:R2-final} and \eqref{eq:R3-final} gives us the first bound in Lemma \ref{lem:Rj}.

\medskip

\noindent
$\boxed{j=4}$ The simple estimate $0\le R_4\le N^{3\beta-1}\cN^2$ follows from the uniform bound $0\le w_N\le CN^{3\beta}$. Moreover, by Sobolev's inequality, we have
$$ w_N(x-y) \le C\|w_N\|_{L^{3/2}} (-\Delta_x) \le CN^{\beta } (-\Delta_x-\Delta_y)$$
as quadratic form on $\gH^2$ (see  e.g. \cite[Lemma 3.2]{NamRouSei-15} for a proof). Therefore,
$$
R_4 \le CN^{\beta-1} \iint (-\Delta_x -\Delta_y) a_x^* a_y^* a_x a_y \d x \d y \le CN^{\beta-1} \dGamma(-\Delta) \cN.
$$
This completes the proof of Lemma \ref{lem:Rj}.  
\end{proof}

Heuristically, the first estimate in Lemma \ref{lem:Rj} tells us that $R_4$ is the main error term among all $R_j$'s. The simple bound $0\le R_4\le N^{3\beta-1}\cN^2$ can serve as a-priori estimate, but it is not sufficient when $\beta>1/3$. On the other hand, in order to use the bound $R_4\le CN^{\beta-1} \dGamma(-\Delta) \cN$ we need to control the kinetic energy $\langle \Phi_N(t), \dGamma(1-\Delta) \Phi_N(t)\rangle$. We have

\begin{lemma} \label{lem:HN-kinetic} Under the assumptions in Theorem \ref{thm:main}, we have  
$$
\big \langle \Phi_N(t), \dGamma(1-\Delta)  \Phi_N(t) \big\rangle \le  C_\eps  N^{\beta+\eps} , \quad \forall t>0, \,\forall \eps\in (0,1-2\beta].
$$
\end{lemma}

The proof of Lemma \ref{lem:HN-kinetic} is similar to that of Lemma \ref{lem:bH-kinetic}. We will need 

\begin{lemma} \label{lem:dt-Rj} We have the quadratic form estimates on $\cF_+^{\le N}$: 
\begin{align*}
\pm \partial_t (R_j+R_j^*) \le \eta \Big( R_4 + \frac{\cN^2}{N} \Big)+ \frac{C (1+\cN)}{\eta(1+t)^3}, \\
\pm i[(R_j+R_j^*),\cN] \le \eta \Big( R_4 + \frac{\cN^2}{N} \Big)+ \frac{C (1+\cN)}{\eta(1+t)^3},
\end{align*}
for all $j=0,1,2,3,4$ and $\eta>0$. The constant $C$ depends only on $\|u(0,\cdot)\|_{W^{1,\ell}}$.
\end{lemma}

\begin{proof} First, we bound $i[(R_j+R_j^*),\cN]$. If $j=0$ or $j=4$, the commutator is $0$. Moreover, we have
$$ i[R_1,\cN]=iR_1,\quad i[R_2,\cN]=-2iR_2, \quad i[R_3,\cN]=iR_3$$
because $[a_x,\cN]=a_x$, $[a_x^*a_y^*,\cN]=-2 a_x^* a_y^*$ and $[a_x^*a_y a_z, \cN]=a^*_x a_y a_z$, respectively. Thus the desired inequalities can be obtained in the same way as in Lemma \ref{lem:Rj} (more precisely, they follow from \eqref{eq:R1-final-0}, \eqref{eq:R2-final-0} and \eqref{eq:R3-final-0}). 

Next, we bound $\partial_t (R_j+R_j^*)$ by proceeding as in the proof of Lemma \ref{lem:Rj}. Let us explain term by term. 

\medskip

\noindent
$\boxed{j=0}$ From \eqref{eq:K1-norm-1}-\eqref{eq:K1-norm-2}-\eqref{eq:K1-norm} and  \eqref{eq:der-K1-1}-\eqref{eq:der-K1-2}-\eqref{eq:der-K1}, we find that
$$ \Big\| \partial_t \Big( Q(t)[w_N*|u(t)|^2+\widetilde{K}_1(t) -\mu_N(t)]Q(t) \Big) \Big \| \le  \frac{C}{(1+t)^3}.$$
Therefore, similarly to \eqref{eq:R0-final}, we have 
\begin{align}\label{eq:dt-R0-final}
\pm \partial_t R_0 &= \pm \d\Gamma\Big(\partial_t \Big( Q(t)[w_N*|u(t)|^2+\widetilde{K}_1(t)-\mu_N(t)]Q(t) \Big) \Big)\frac{1-\cN}{N-1}  \\
& \le \frac{C\cN^2}{N(1+t)^{3}} \le \eta \frac{\cN^2}{N} + \eta^{-1} \frac{C\cN^2}{N(1+t)^6} \le \eta \frac{\cN^2}{N} + \eta^{-1} \frac{C\cN}{(1+t)^6}, \quad \forall \eta>0 .\nn
\end{align}

\noindent 
$\boxed{j=1}$ Using $\|Q(t)\|\le 1$, $\|\partial_t Q(t)\| \le C$, \eqref{eq:K1-norm-1} and \eqref{eq:der-K1-1}, we have
\begin{align*}
&\Big\|\partial_t \Big( Q(t)[w_N*|u(t)|^2]u(t) \Big)\Big\|_{L^2}  \\
&\le C \| [w_N*|u(t)|^2] u(t)\|_{L^2} +  \| [\partial_t[ (w_N* |u(t)|^2)  u(t)]\|_{L^2} \le \frac{C}{(1+t)^{3}}.
\end{align*}
Therefore, we can follow the proof of \eqref{eq:R1-final} and obtain
\bq \label{eq:dt-R1-final}
\pm \partial_t (R_1 + R_1^*) \le  \eta \frac{\cN^2}{N} + \eta^{-1} \frac{C\cN}{(1+t)^3}, \quad \forall \eta>0.
\eq

\noindent
$\boxed{j=2}$ For every $\Phi\in \cF_+^{\le N}(t)$, we have
\begin{align*}
&\langle \Phi, \partial_t R_2  \Phi \rangle =\iint \Big[ (\partial_t Q(t) \otimes 1 + 1\otimes \partial_t Q(t)) \widetilde K_2(t,x,y) + \partial_t \widetilde K_2(t,x,y)\Big] \times \\
&\qquad \qquad \qquad \times   \Big \langle \Phi, a^*_x a^*_y \Big(\frac{\sqrt{(N-\cN)(N-\cN-1)}}{N-1}-1\Big) \Phi \Big\rangle \d  x \d y.
\end{align*}
Here we have used the decomposition
$$
\partial_t K_2(t)= \partial_t Q(t) \otimes Q(t) \widetilde K_2(t) + Q(t) \otimes \partial_t Q(t)  \widetilde K_2(t) + Q(t)\otimes Q(t) \partial_t \widetilde K_2(t)
$$
and omitted the projection $Q(t)$ using $\Phi \in \cF_+(t)$. Similarly to \eqref{eq:R2-final-0}, we have   
\begin{align*} 
& \left|  \iint   \partial_t \widetilde K_2(t,x,y) \Big \langle \Phi, a^*_x a^*_y  \Big (\frac{\sqrt{(N-\cN)(N-\cN-1)}}{N-1}-1 \Big) \Phi \Big\rangle \d  x \d y \right|  \nn\\
&\le  \eta \langle \Phi, R_4 \Phi\rangle + \frac{C}{\eta (1+t)^3} \langle \Phi, (1+\cN) \Phi\rangle, \quad \forall \eta>0.
\end{align*}
The term involving $(\partial_t Q(t) \otimes 1 + 1 \otimes \partial_t Q(t)) \widetilde K_2(t,x,y)$ is bounded as
\begin{align*}
&\left| \iint \big(\partial_t Q(t) \otimes 1+1 \otimes \partial_t Q(t) \big) \widetilde K_2(t,x,y)  \times \right. \\
&\qquad \qquad \times \left.   \Big \langle \Phi, a^*_x a^*_y \Big(\frac{\sqrt{(N-\cN)(N-\cN-1)}}{N-1}-1\Big) \Phi \Big\rangle \d  x \d y \right| \\
&\le \left( \iint \Big| \big(\partial_t Q(t) \otimes 1 +1 \otimes \partial_t Q(t) \big) \widetilde K_2(t,x,y) \Big|^2 \d x \d y \right)^{1/2}\times  \nn\\
& \qquad\qquad \times \left( \iint \| a_x a_y \Phi\|^2 \d x \d y \right)^{1/2} \left\| \left(\frac{\sqrt{(N-\cN)(N-\cN-1)}}{N-1}-1\right) \Phi  \right\| \nn \\
&\le \frac{C}{\sqrt{N}(1+t)^{3/2}} \langle \Phi, \cN^2 \Phi \rangle^{1/2}\langle \Phi, (1+\cN)\Phi \rangle^{1/2}.
\end{align*}
Here we have used \eqref{eq:L2-dtQ-K2} and \eqref{eq:sqrt-N-N-1} in the last estimate. In summary, 
$$
\left| \langle \Phi, \partial_t R_2  \Phi \rangle  \right| \le  \eta \Big( \langle \Phi, R_4 \Phi\rangle + \frac{1}{N}\langle \Phi, \cN^2 \Phi\rangle\Big)  + \frac{C}{\eta (1+t)^3} \langle \Phi, (1+\cN) \Phi\rangle
$$
for all $\Phi\in \cF_+^{\le N}(t)$ and $\eta>0$. Therefore,
\bq \label{eq:dt-R2-final}
\pm \partial_t (R_2 + R_2^*) \le \eta \Big( R_4 +\frac{\cN^2}{N} \Big)+ \frac{C(1+\cN)}{\eta(1+t)^{3}}, \quad \forall \eta>0.
\eq
\medskip

\noindent
$\boxed{j=3}$ For all $\Phi\in \cF_+^{\le N}(t)$, we have 
\begin{align*}
\langle \Phi, \partial_t R_3 \Phi\rangle & = \frac{1}{N-1}  \iiiint \Big[   \big( 1 \otimes Q(t) w_N Q(t)\otimes Q(t)\big)(x,y;x',y') \overline{\partial_t u(t,x)} \\
&\qquad \qquad  +  \Big( \partial_t \big( 1 \otimes Q(t) w_N Q(t)\otimes Q(t)\big)\Big)(x,y;x',y') \overline{ u(t,x)} \Big] \times \\
& \qquad \qquad \qquad \qquad\qquad  \times \langle \Phi, \sqrt{N-\cN} a^*_y a_{x'} a_{y'}  \Phi \rangle \,\d x \d y \d x' \d y'.
\end{align*}
The term involving $\partial_t u(t,x)$ can  be estimated similarly to  \eqref{eq:R3-final-0}:
\begin{align*}
&\frac{1}{N-1} \left| \iiiint  ( 1 \otimes Q(t) w_N Q(t)\otimes Q(t))(x,y;x',y') \overline{\partial_t  u(t,x)}  \times \right. \\
& \qquad \qquad \qquad \qquad \left. \times  \langle \Phi, \sqrt{N-\cN} a^*_y a_{x'} a_{y'}  \Phi \rangle \,\d x \d y \d x' \d y' \right| \\
&\le \eta \langle \Phi, R_4 \Phi \rangle + \frac{C}{\eta (1+t)^3} \langle \Phi, \cN \Phi \rangle, \quad \forall \eta>0.
\end{align*}

In the following, we will use the kernel estimate
\bq \label{eq:kernel-dtQ}
|(\partial_t Q(t))(z;z')|= |\partial_t u(t,z) \overline{u(t,z')}+ u(t,z) \overline{\partial_t u(t,z')}| \le q(z) q(z')
\eq
where $q(t,z):=|u(t,z)| + |\partial_t u(t,z)|.$ Recall that by Lemma \ref{lem:Hartree-equation}, 
$$
\|q(t, \cdot)\|_{L^2} \le C, \quad \|q(t,\cdot)\|_{L^\infty}  \le \frac{C}{(1+t)^{3/2}}.
$$

Let us decompose  $\partial_t \big( 1 \otimes Q(t) w_N Q(t)\otimes Q(t)\big)$ into three terms. For the first term $1 \otimes \partial_t Q(t) w_N Q(t)\otimes Q(t)$, we can estimate 
\begin{align*}
&\frac{1}{N-1} \left| \iiiint  (1 \otimes \partial_t Q(t) w_N Q(t)\otimes Q(t)) (x,y;x',y')  \overline{u(t,x)}  \times \right.\\
& \qquad \qquad \qquad \qquad \qquad \qquad\left. \times \Big\langle \Phi, \sqrt{N-\cN} a^*_y a_{x'} a_{y'}  \Phi \Big\rangle \,\d x \d y \d x' \d y' \right|\\
&=\frac{1}{N-1}  \left|  \iiiint  (\partial_t Q(t))(y;y') w_N(x-y') \delta(x-x')  \overline{ u(t,x)} \times \right.\\
&\qquad \qquad \qquad \qquad \qquad \qquad  \left.\times \Big \langle \Phi,\sqrt{N-\cN} a^*_y a_x a_{y'}  \Phi \Big\rangle \,\d x \d y \d x' \d y' \right| \nn \\
&\le \frac{1}{N-1}  \iiint    q(t,y) q(t,y') w_N(x-y')  |u(t,x)| \times \\
&\qquad \qquad \qquad \qquad \qquad \qquad  \times \| a_y \sqrt{N-\cN} \Phi\| \|a_x a_{y'}  \Phi \|\,\d x \d y \d y' \nn \\
&\le \frac{\|q(t,\cdot)\|_{L^\infty}}{N-1} \left( \int |q(t,y)|^2 \d y\right)^{1/2} \left( \int \| a_y \sqrt{N-\cN} \Phi \|^2 \d y\right)^{1/2} \times \\
& \times \left( \iint w_N(x-y') |u(t,x)|^2 \d x \d y'\right)^{1/2} \left( \iint w_N(x-y') \| a_x a_{y'} \Phi\|^2 \d x \d y'\right)^{1/2} \\
&\le \frac{C}{(1+t)^{3/2}} \langle \Phi, R_4 \Phi \rangle^{1/2} \langle \Phi, \cN \Phi\rangle^{1/2} \le \eta \langle \Phi, R_4 \Phi \rangle + \frac{C\langle \Phi, \cN \Phi \rangle}{\eta (1+t)^3} , \quad \forall \eta>0.
\end{align*}
For the second term $1 \otimes  Q(t) w_N  \partial_t Q(t) \otimes Q(t)$, we have 
\begin{align*}
&\frac{1}{N-1} \left| \iiiint  (1 \otimes  Q(t) w_N  \partial_t Q(t) \otimes Q(t))(x,y;x',y')   \overline{u(t,x)}  \times \right. \\
& \qquad \qquad \qquad \qquad \qquad \qquad \left. \times \Big \langle \Phi, \sqrt{N-\cN} a^*_y a_{x'} a_{y'}  \Phi \Big\rangle \,\d x \d y \d x' \d y' \right|\\
&=  \frac{1}{N-1}  \left| \iiiint   w_N(x-y) (\partial_tQ(t))(x,x')  \delta(y-y') \overline{ u(t,x)} \times \right.\\
&\qquad \qquad \qquad \qquad \qquad \qquad \left.\times \Big \langle \Phi,\sqrt{N-\cN} a^*_y a_{x'} a_{y'}  \Phi \Big\rangle \,\d x \d y \d x' \d y' \right| \nn \\
&\le \frac{1}{N-1}   \iiint   w_N(x-y) q(t,x) q(t,x') |u(t,x)|  \times \\
&\qquad \qquad \qquad \qquad\qquad \qquad  \times \| a_y \sqrt{N-\cN} \Phi\| \|a_{x'} a_{y}  \Phi\| \,\d x \d y \d x' \nn \\
&\le \frac{1}{N-1} \|w_N\|_{L^1}  \| q(t,\cdot)\|_{L^\infty}  \|u(t,\cdot)\|_{L^\infty} \times \\
&\quad \times \left( \iint \|a_{x'} a_{y}  \Phi\|^2 \,\d x' \d y \right)^{1/2}  \left( \iint |q(t,x')|^2 \| a_y \sqrt{N-\cN} \Phi\|^2 \d x' \d y \right)^{1/2} \\
&\le \frac{C}{\sqrt{N}(1+t)^3}  \langle \Phi, \cN^2 \Phi\rangle^{1/2} \langle \Phi, \cN \Phi\rangle^{1/2}  \le \eta \frac{\langle \Phi, \cN^2 \Phi\rangle}{N} + \frac{C\langle \Phi, \cN \Phi\rangle}{\eta(1+t)^3}, \,\, \forall \eta>0.
\end{align*}
The third term $1 \otimes  Q(t) w_N   Q(t) \otimes \partial_t Q(t)$ is bounded similarly. Thus 
\begin{align*}
\left| \langle \Phi, \partial_t R_3 \Phi \rangle\right| \le \eta \left( \langle \Phi, R_4 \Phi \rangle + \frac{\langle \Phi, \cN^2 \Phi\rangle}{N} \right)  + \frac{C\langle \Phi, \cN \Phi\rangle}{\eta(1+t)^3}
\end{align*}
for all $\Phi\in \cF_+^{\le N}(t)$ and $\eta>0$. Consequently,
\bq \label{eq:dt-R3-final}
\pm  \partial_t (R_3+R_3^*) \le \eta \left( R_4 +\frac{\cN^2}{N} \right)+  \frac{C\cN}{\eta(1+t)^3}, \quad \forall \eta>0.
\eq

\noindent
$\boxed{j=4}$ For all $\Phi\in \cF_+^{\le N}(t)$, we have 
\begin{align*}
&\langle \Phi, \partial_t R_4\Phi\rangle = \frac{1}{2(N-1)} \Re \iiiint \partial_t \Big ( Q(t) \otimes Q(t) w_N Q(t)\otimes Q(t)\Big)(x,y;x',y') \\
& \qquad \qquad \qquad \qquad \qquad \qquad\qquad \qquad\times \langle \Phi, a_x^* a^*_y a_{x'} a_{y'}  \Phi \rangle \,\d x \d y \d x' \d y'.
\end{align*}
Let us decompose $\partial_t \big ( Q(t) \otimes Q(t) w_N Q(t)\otimes Q(t)\big)$ into four terms, and consider for example $\partial_t Q(t) \otimes Q(t) w_N Q(t)\otimes Q(t)$. Using \eqref{eq:kernel-dtQ} again, we have
\begin{align*}
&\frac{1}{N-1} \left|  \iiiint \Big(\partial_t Q(t) \otimes Q(t) w_N Q(t)\otimes Q(t)\Big)(x,y;x',y')\times \right.\\
& \qquad \qquad \qquad \qquad \qquad \qquad \qquad \qquad \left. \times \Big\langle \Phi, a_x^* a^*_y a_{x'} a_{y'}  \Phi \Big\rangle \,\d x \d y \d x' \d y' \right|\\
&=  \frac{1}{N-1} \left| \iiiint (\partial_t Q(t))(x,x') w_N(x'-y) \delta(y-y')\times \right. \\
& \qquad \qquad \qquad \qquad \qquad \qquad\qquad \qquad \left. \times \Big\langle \Phi, a_x^* a^*_y a_{x'} a_{y'}  \Phi \Big\rangle \,\d x \d y \d x' \d y' \right|\\
&\le \frac{1}{N-1} \iiint q(t,x) q(t,x') w_N(x'-y) \| a_x a_y \Phi\| \|a_{x'} a_{y} \Phi\| \,\d x \d y \d x'\\
&\le \frac{\|q(t,\cdot)\|_{L^\infty}}{N-1} \left( \iiint w_N(x'-y) \| a_x a_y \Phi\|^2 \d x \d y \d x'  \right)^{1/2} \times \\
&\qquad \qquad \qquad \qquad  \times \left( \iiint w_N(x'-y) \| a_{x'} a_y \Phi\|^2 |q(t,x)|^2 \d x \d y \d x'  \right)^{1/2} \\
&\le \frac{C}{N^{1/2}(1+t)^{3/2}}   \langle \Phi, \cN^2 \Phi\rangle^{1/2} \langle \Phi, R_4 \Phi\rangle^{1/2} \\
& \le \eta \langle \Phi, R_4 \Phi\rangle + \frac{C\langle \Phi, \cN^2 \Phi\rangle}{\eta N (1+t)^{3} } \le \eta \langle \Phi, R_4 \Phi\rangle + \frac{C\langle \Phi, \cN \Phi\rangle}{\eta (1+t)^{3} }, \quad \forall \eta>0.
\end{align*}
By similar estimates, we find that
\begin{align*}
\left| \langle \Phi, \partial_t R_4\Phi \rangle \right| \le \eta \langle \Phi, R_4 \Phi\rangle + \frac{C\langle \Phi, \cN \Phi\rangle}{\eta (1+t)^{3} }
\end{align*}
for all $\Phi\in \cF_+^{\le N}(t)$ and $\eta>0$. Thus 
\bq \label{eq:dt-R4-final}
\pm \partial_t R_4 \le \eta R_4 + \frac{C\cN}{\eta(1+t)^3}, \quad \forall \eta>0.
\eq
This completes the proof. 
\end{proof}

Now we are ready to provide
\begin{proof}[Proof of Lemma \ref{lem:HN-kinetic}] We use the proof strategy of Lemma \ref{lem:bH-kinetic}. Using the equation \eqref{eq:eq-PhiNt} we can write 
\begin{align} \label{eq:PhiN-kin-0}
\big\langle \Phi_N(t), \widetilde H_N(t) \Phi_N(t) \big\rangle &- \big\langle \Phi_N(0), \widetilde H_N(0) \Phi_N(0) \big\rangle \nn\\
&= \int_0^t   \big\langle \Phi_N(s), \partial_s \widetilde H_N(s) \Phi_N(s) \big\rangle \d s
\end{align}
and
\begin{align} \label{eq:PhiN-kin-a}
\big\langle \Phi_N(t), \cN  \Phi_N(t) \big\rangle - \big\langle \Phi_N(0), \cN \Phi_N(0) \big\rangle = \int_0^t   \big\langle \Phi_N(s), i[\widetilde H_N(s),\cN] \Phi_N(s) \big\rangle \d s.
\end{align}

Let us estimate both sides of \eqref{eq:PhiN-kin-0}. Recall that
$$
\widetilde H_N (t)=  \1^{\le N} \Big[ \bH(t) + \frac{1}{2}\sum_{j=0}^4 ( R_{j} + R_j^*) \Big] \1^{\le N}.
$$
From Lemma \ref{lem:bHt-dbHt} and Lemma \ref{lem:Rj}, we have the form estimates on $\cF_+^{\le N}(t)$:
\begin{align} \label{eq:HNt-upper-lower}
&\pm \1^{\le N} \Big( \widetilde H_N(t) + \dGamma(\Delta) -  R_4 \Big) \1^{\le N}  \nn\\
&= \pm \1^{\le N}  \Big( \bH(t) + \dGamma(\Delta) + \frac{1}{2}\sum_{j=0}^3 ( R_{j} + R_j^*)   \Big) \1^{\le N}  \nn\\
&\le \eta\Big(\dGamma(1-\Delta)+R_4+\frac{\cN^2}{N} \Big) +  \frac{C_\eps (N^{\beta+\eps} + \cN)}{\eta(1+t)^3}.
\end{align}
Similarly, from Lemma \ref{lem:bHt-dbHt} and Lemma \ref{lem:dt-Rj}, we have 
\begin{align} 
\pm \partial_t \widetilde H_N(t) &\le \eta\Big(\dGamma(1-\Delta)+R_4 + \frac{\cN^2}{N}\Big) +  \frac{C_\eps (N^{\beta+\eps} + \cN)}{\eta(1+t)^3},  \label{eq:dt-HNt-upper-lower}
\end{align}
for all $\eta>0$. 

Applying \eqref{eq:HNt-upper-lower} with $\eta=1/2$ and using $\cN^2/N \le \cN \le \dGamma(1-\Delta)$ on $\cF_+^{\le N}(t)$, we find that
\begin{align} \label{eq:PhiN-kin-1}
\big\langle \Phi_N(t), \widetilde H_N(t) \Phi_N(t) \big\rangle & \ge \frac{1}{2} \big\langle \Phi_N(t), (\dGamma(1-\Delta) + R_4) \Phi_N(t) \big\rangle \nn\\
&\qquad\qquad- C_\eps \Big(N^{\beta+\eps} + \big\langle \Phi_N(t), \cN  \Phi_N(t) \big\rangle \Big),\\
\label{eq:PhiN-kin-2-00}
\big\langle \Phi_N(0), \widetilde H_N(0) \Phi_N(0) \big\rangle  &\le C \big\langle \Phi_N(0), (\dGamma(1-\Delta) + R_4) \Phi_N(0) \big\rangle  + C_\eps N^{\beta+\eps}.
\end{align}
Using $\Phi_N(0)=\1^{\le N}\Phi(0)$,  we get
$$ \big\langle \Phi_N(0), \dGamma(1-\Delta)  \Phi_N(0) \big\rangle  \le \big\langle \Phi(0), \dGamma(1-\Delta)  \Phi(0) \big\rangle \le C_\eps N^{\beta+\eps}.  $$
On the other hand, recall that $R_4\le CN^{3\beta-1}\cN^2$ by Lemma \ref{lem:Rj}. Moreover, it is well-known that for every quasi-free state $\Phi$ and $s\in \mathbb{N}$, we have
\bq \label{eq:moment-quasi-free}
\Big\langle \Phi, (1+\cN)^{s} \Phi \Big\rangle \le C_{s} \Big\langle \Phi, (1+\cN) \Phi \Big\rangle^{s}
\eq
where the constant $C_s$ depends only on $s$ (see \cite[Lemma 5]{NamNap-15} for a proof). Combining with the assumptions $\big\langle \Phi(0), \cN \Phi(0)\big \rangle\le C_\eps N^\eps$ and $\eps \le 1-2\beta$, we obtain
\begin{align*}
\big\langle \Phi_N(0), R_4   \Phi_N(0) \big\rangle &\le CN^{3\beta-1} \big\langle \Phi(0), \cN^2 \Phi(0) \big\rangle \le CN^{3\beta-1} \big\langle \Phi(0), \cN \Phi(0)\big \rangle^2  \\
& \le C_\eps N^{3\beta-1} N^{2\eps} \le C_\eps N^{\beta+\eps}.
\end{align*}
 Thus \eqref{eq:PhiN-kin-2-00} reduces to   
\begin{align}  \label{eq:PhiN-kin-2}
\big\langle \Phi_N(0), \widetilde H_N(0) \Phi_N(0) \big\rangle \le C_\eps N^{\beta+\eps}.
\end{align}

Next, we apply \eqref{eq:dt-HNt-upper-lower} with $\eta=(1+t)^{-3/2}$ and use $\cN^2/N\le \cN \le \dGamma(1-\Delta)$ on $\cF_+^{\le N}(t)$. This gives
\begin{align} \label{eq:PhiN-kin-3}
\big\langle \Phi_N(t), \partial_t \widetilde H_N(t) \Phi_N(t) \big\rangle & \le C_\eps \frac{\big\langle \Phi_N(t), (\dGamma(1-\Delta) + R_4) \Phi_N(t) \big\rangle+N^{\beta+\eps}}{(1+t)^{3/2}}.
\end{align}
Inserting \eqref{eq:PhiN-kin-1}, \eqref{eq:PhiN-kin-2} and \eqref{eq:PhiN-kin-3} into \eqref{eq:PhiN-kin-0} we obtain
\begin{align} \label{eq:PhiN-kin-4}
\big\langle \Phi_N(t), (\dGamma(1-\Delta) + R_4) \Phi_N(t) \big\rangle &\le C_\eps \int_0^t  \frac{\big\langle \Phi_N(s), (\dGamma(1-\Delta) + R_4) \Phi_N(s) \big\rangle}{(1+s)^{3/2}} \d s \nn 
\\
&+ C_\eps \Big( N^{\beta+\eps}  +  \big\langle \Phi_N(t), \cN  \Phi_N(t) \big\rangle \Big).
\end{align}

Now we consider \eqref{eq:PhiN-kin-a}. By using $\Phi(0)=\1^{\le N}\Phi(0)$ and the assumption on $\Phi(0)$, we have
$$ \langle \Phi_N(0), \cN \Phi_N(0) \rangle \le \langle \Phi(0), \cN \Phi(0) \rangle \le C_\eps N^{\beta+\eps}.$$
Moreover, from Lemma \ref{lem:bHt-dbHt} and Lemma \ref{lem:dt-Rj}, we have
$$
\pm i [\widetilde H_N(t), \cN] \le \eta\Big(\dGamma(1-\Delta)+R_4 + \frac{\cN^2}{N}\Big) +  \frac{C_\eps (N^{\beta+\eps} + \cN)}{\eta(1+t)^3}, \quad \forall \eta>0.
$$
We can choose $\eta=(1+t)^{-3/2}$ and use $\cN^2/N \le \cN \le \dGamma(1-\Delta)$ on $\cF_+^{\le N}(t)$ to obtain
$$
\pm i [\widetilde H_N(t), \cN] \le C_\eps \frac{\dGamma (1-\Delta)+R_4 + N^{\beta+\eps}}{(1+t)^{3/2}}.
$$
Inserting the latter estimate into the right side of \eqref{eq:PhiN-kin-a}, we get
\bq \label{eq:PhiN-kin-b}
\langle \Phi(t), \cN \Phi(t) \rangle \le C_\eps \int_0^t  \frac{\big\langle \Phi_N(s), (\dGamma(1-\Delta) + R_4) \Phi_N(s) \big\rangle}{(1+s)^{3/2}} \d s + C_\eps N^{\beta+\eps}.
\eq

Finally, we substitute \eqref{eq:PhiN-kin-b} into the right side of \eqref{eq:PhiN-kin-4} and find that
\begin{align} \label{eq:PhiN-kin-6}
&\big\langle \Phi_N(t), (\dGamma(1-\Delta) + R_4) \Phi_N(t) \big\rangle \nn\\
&\le C_\eps \int_0^t  \frac{\big\langle \Phi_N(s), (\dGamma(1-\Delta) + R_4) \Phi_N(s) \big\rangle}{(1+s)^{3/2}} \d s + C_\eps N^{\beta+\eps}  .
\end{align}
This bound is similar to \eqref{eq:Phit-Grw-1} and we can argue as in the proof of Lemma \ref{lem:bH-kinetic} to conclude that 
$$
\big\langle \Phi_N(t), (\dGamma(1-\Delta) + R_4) \Phi_N(t) \big\rangle \le C_\eps N^{\beta+\eps}.
$$ 
Since $R_4\ge 0$, the desired kinetic estimate follows. 
\end{proof}

\section{Proof of the main theorem} \label{sec:main-proof}

\begin{proof}[Proof of Theorem \ref{thm:main}] It suffices to consider the case when $N$ is large and $\eps$ is small (indeed, if the desired bound holds for some $\eps>0$, then it also holds for any $\eps'>\eps$ because $(1+t)^{1+\eps} N^{(2\beta+\eps-1)/2} \le (1+t)^{1+\eps'} N^{(2\beta+\eps'-1)/2}$). In particular, we will assume $\eps<\min \{1/2,1-2\beta\}$. 

Since $U_N(t)$ is a unitary operator from $\gH^N$ to $\cF_+^{\le N}(t)\subset \cF(\gH)$, we have
\begin{align} \label{eq:final-proof-1}
\|  \Psi_N(t) - U_N(t)^* \1^{\le N} \Phi(t)\|_{\gH^N} &= \| U_N(t) \Psi_N(t) - \1^{\le N} \Phi(t)\| \\
&=  \| \1^{\le N}( \Phi_{N}(t) - \Phi(t) )\| \le \| \Phi_{N}(t)- \Phi(t)\|. \nn
\end{align} 
Using the equations \eqref{eq:Bogoliubov-equation} and \eqref{eq:eq-PhiNt}, we can compute 
\begin{align*} 
 \partial_t \| \Phi_{N}(t)-\Phi(t)\|^2 &= - 2\Re \, \partial_t \big\langle \Phi_{N}(t), \Phi(t) \big\rangle  \nn\\
 &=  - 2\Re \Big( \big\langle \partial_t \Phi_N(t), \Phi(t) \big\rangle +\big \langle  \Phi_N(t), \partial_t \Phi(t)\big \rangle \Big)  \nn \\
 & = - 2\Re \Big( \big\langle -i \widetilde H_N(t) \Phi_N(t), \Phi(t)\big \rangle +\big \langle  \Phi_N(t), -i \bH(t) \Phi(t)\big \rangle \Big)  \nn \\
 & = 2\Re \,\big \langle i \Phi_N(t),  (\widetilde H_N(t) - \bH(t))\Phi(t)\big\rangle .
 \end{align*}
Since 
$$
\widetilde H_N (t)=  \1^{\le N} \Big[ \bH(t) + \frac{1}{2}\sum_{j=0}^4 ( R_{j} + R_j^*) \Big] \1^{\le N}
$$
and $\Phi_N \in \cF_+^{\le N}(t)$, we can define $\1^{>N} := \1 - \1^{\le N}$ and write
\begin{align}  \label{eq:final-proof-2}
 \partial_t \| \Phi_{N}(t)-\Phi(t)\|^2 &=   \sum_{j=0}^4 \Re \big\langle i\Phi_N(t), (R_j+R_j^*) \1^{\le N} \Phi(t) \big\rangle\nn\\
 &\qquad \qquad \qquad - 2\Re\,\big\langle i\Phi_N(t), \bH \1^{>N} \Phi(t) \big\rangle.
\end{align}

\medskip

\noindent
{\bf Step 1.} Let us consider the last term of \eqref{eq:final-proof-2}. Since $\Phi_N(t) \in \cF_+^{\le N}(t)$ and $\1^{\le N}\dGamma(h) \1^{>N}=0$, we have 
$$
\big\langle \Phi_N(t), \bH \1^{>N} \Phi(t) \big\rangle = \big\langle \Phi_N(t), (\bH - \dGamma(h))  \1^{>N} \Phi(t) \big\rangle.
$$
Recall that by \eqref{eq:bound-paring-dG1},
\bq \label{eq:final-proof-simplebH}
\pm (\bH - \dGamma(h)) \le C( \cN + N^{3\beta}).
\eq
Here $C$ is a general constant depending only on $\|u(0,\cdot)\|_{W^{\ell,1}}$ (more precisely, on $\kappa_0$ in the condition $\| u(0,\cdot)\|_{W^{\ell,1}} \le \kappa_0$).

We will use the general fact that if $A$ and $B$ are quadratic forms satisfying  $\pm B\le A$, then for all vectors $f,g$ we have the Cauchy-Schwarz type inequality 
\begin{align} \label{eq:final-proof-CS}
|\langle f, B g\rangle| & \le | \langle f, (A+B) g \rangle | +  |\langle f, A g\rangle | \nn\\
&\le \langle f, (A+B) f \rangle^{1/2} \langle g, (A+B) g \rangle^{1/2} +  \langle f, A f\rangle^{1/2}\langle g, A g\rangle^{1/2} \nn\\
&\le 3 \langle f, A f\rangle^{1/2}\langle g, A g\rangle^{1/2}.
\end{align}
From \eqref{eq:final-proof-simplebH} and  \eqref{eq:final-proof-CS}, we obtain
\begin{align}  \label{eq:final-proof-bH>N}
&\left| \big\langle \Phi_N(t), \bH \1^{>N} \Phi(t) \big\rangle \right| = \left| \big\langle \Phi_N(t), (\bH - \dGamma(h))  \1^{>N} \Phi(t) \big\rangle \right| \nn \\
&\le C \Big\langle \Phi_N (t), ( \cN + N^{3\beta}) \Phi_N(t) \Big\rangle^{1/2} \Big\langle \1^{>N}\Phi (t), ( \cN + N^{3\beta}) \1^{>N}\Phi(t) \Big\rangle^{1/2} \nn\\
& \le C (N+N^{3\beta}) \Big\langle \Phi (t),  \cN^s N^{-s} \Phi(t) \Big\rangle^{1/2}, \quad \forall s \ge 1. 
\end{align}
Here, in the last inequality, we have used $\cN \le N$ on $\cF_+^{\le N}(t)$ and 
$$\1^{>N} (\cN + N^{3\beta}) \le (N+N^{3\beta})  \cN^{s}N^{-s}, \quad \forall s\ge 1.$$
Now we use the moment estimate \eqref{eq:moment-quasi-free}, the bound on $\langle \Phi(t), \cN \Phi(t)\rangle$ in Lemma \ref{lem:Bogoliubov-equation} and the assumption $\langle \Phi(0), \cN \Phi(0)\rangle \le C_\eps N^{\eps}$. All this gives
\begin{align} \label{eq:final-proof-moment}
\big\langle \Phi(t), (1+\cN)^{s} \Phi(t) \big\rangle \le C_s \big\langle \Phi(t), (1+\cN) \Phi(t) \big\rangle^{s} \le C_{\eps,s} N^{2s\eps} [\log(2+t)]^{2s}.
\end{align}
Hence, \eqref{eq:final-proof-bH>N} reduces to 
\begin{align*}
\left| \big\langle \Phi_N(t), \bH \1^{>N} \Phi(t) \big\rangle \right| \le C_{\eps,s} (N+N^{3\beta}) N^{s(\eps-1/2)}  [\log(2+t)]^{s}
\end{align*}
for all $s\ge 1$. Since $\eps-1/2<0$, we can choose $s=s(\eps)$ sufficiently large (e.g. $s \ge (2+3\beta)/(1/2-\eps)$) to obtain 
\begin{align} \label{eq:final-proof-bH>N-last}
\left| \big\langle \Phi_N(t), \bH \1^{>N} \Phi(t) \big\rangle \right| &\le C_{\eps} N^{-1} (1+t)^\eps. 
\end{align}
Here we have bound $[\log(2+t)]^{s_\eps}$ by $C_\eps (1+t)^\eps$ for simplicity.  
\medskip

\noindent{\bf Step 2.} Now we turn to the first term on the right side of \eqref{eq:final-proof-2}. Recall that by Lemma \ref{lem:Rj}, we have the quadratic form estimates on $\cF_+^{\le N}(t)$:
\bq \label{eq:final-proof-Rj}
\pm (R_j+R_j^*) \le 2(1+ \eta) R_4 + \eta \frac{\cN^2}{N} + \frac{C (1+\cN)}{\eta(1+t)^3}
\eq
for all $j=0,1,2,3,4$ and $\eta> 0$ (the bound for $j=4$ does not follow from Lemma \ref{lem:Rj} but it is trivial).

Since we do not have a good control on $\langle \Phi_N(t), \cN \Phi_N(t)\rangle$, we need to introduce a cut-off before applying \eqref{eq:final-proof-Rj}. Note that for every $4<M<N-2$, 
$$ \1^{\le M} (R_j+R_j^*)\1^{> M+2} =0 \quad\text{and}\quad  \1^{> M} (R_j+R_j^*)\1^{\le M-2} =0$$
because there are at most 2 creation or annihilation operators in the expressions of $R_j$'s.  Therefore, we can write
\begin{align*}
\big\langle  \Phi_N(t), (R_{j} +R_j^*)\1^{\le N} \Phi(t) \big\rangle &= \big\langle  \1^{\le M} \Phi_N(t), (R_{j}+R_j^*) \1^{\le M+2} \Phi(t) \big\rangle \\
&\quad + \big\langle  \1^{> M} \Phi_N(t), (R_{j} +R_j^*)  \1^{\le N}  \1^{> M-2}  \Phi(t) \big\rangle\nn
\end{align*}
and then apply \eqref{eq:final-proof-Rj} and \eqref{eq:final-proof-CS} to each term on the right side. This gives 
\begin{align}  \label{eq:final-proof-3}
\left| \big\langle  \Phi_N(t), (R_{j} +R_j^*)\1^{\le N} \Phi(t) \big\rangle \right| \le C(E_1 + E_2) 
\end{align}
where
\begin{align*}
E_1&= \inf_{\eta>0} \left\langle \1^{\le M} \Phi_N(t), \Big( (1+\eta) R_4 + \eta \frac{\cN^2}{N} + \frac{1+\cN}{\eta(1+t)^3} \Big) \1^{\le M} \Phi_N(t)  \right\rangle^{1/2}  \nn\\
& \qquad  \times \left\langle \1^{\le M+2} \Phi(t), \Big( (1+\eta) R_4 + \eta \frac{\cN^2}{N} + \frac{1+\cN}{\eta(1+t)^3} \Big) \1^{\le M+2} \Phi(t)  \right\rangle^{1/2}, \nn \\
E_2& = \inf_{\eta>0} \left\langle \1^{> M} \Phi_N(t), \Big( (1+\eta) R_4 + \eta \frac{\cN^2}{N} + \frac{1+\cN}{\eta(1+t)^3} \Big) \1^{> M} \Phi_N(t)  \right\rangle^{1/2} \\
&\qquad \times \left\langle \1^{> M-2} \Phi(t), \Big((1+\eta) R_4 + \eta \frac{\cN^2}{N} + \frac{1+\cN}{\eta(1+t)^3} \Big) \1^{> M-2} \Phi(t)  \right\rangle^{1/2} .
\end{align*}

To bound $E_1$, we use the last estimate in Lemma \ref{lem:Rj} and $\1^{\le M}\cN \le M$:
$$ \1^{\le M} R_4 \le C N^{\beta-1} \1^{\le M} \cN \dGamma(-\Delta) \le CN^{\beta-1}M \dGamma(-\Delta).$$ 
Moreover, recall that we have the kinetic estimate in Lemma \ref{lem:HN-kinetic}:
$$\langle \Phi_N(t), \dGamma(1-\Delta) \Phi_N(t) \rangle \le C_\eps N^{\beta+\eps}, \quad \forall \eps \in (0,1-2\beta],$$
where the constant $C_\eps$ is independent of $N$ and $t$. Therefore, 
\begin{align*}
&\left\langle \1^{\le M} \Phi_N(t), \Big( (1+\eta) R_4 + \eta \frac{\cN^2}{N} + \frac{1+\cN}{\eta(1+t)^3} \Big) \1^{\le M} \Phi_N(t)  \right\rangle \nn\\
& \le C_\eps \Big( (1+\eta) N^{\beta-1} M N^{\beta+\eps}  + \eta M^2 N^{-1}  + \frac{M}{\eta(1+t)^3} \Big) . 
\end{align*}  
Similarly, we have the same bound with  $\1^{\le M} \Phi_N(t)$ replaced by  $\1^{\le M+2} \Phi(t)$ (the kinetic estimate for $\Phi(t)$ is provided   in Lemma \ref{lem:bH-kinetic}). Then by optimizing over $\eta>0$ we find that 
\begin{align*} 
E_1\le C_\eps  \Big( M N^{(2\beta+\eps-1)/2} + M^{3/2} N^{-1/2} \Big) . 
\end{align*}  

Next, we bound $E_2$ using the argument in Step 1. To be precise, let us choose $\eta=1$ in the variational formula of $E_2$ for simplicity and then use $R_4 \le C N^{3\beta-1}\cN^2$ (see Lemma \ref{lem:Rj}). We obtain the rough bound
\begin{align*}
E_2 & \le C N^{3\beta} \left\langle \1^{> M} \Phi_N(t), (\cN +1)^2  \1^{> M} \Phi_N(t)  \right\rangle^{1/2} \\
&\qquad \qquad \qquad \times \left\langle \1^{> M-2} \Phi(t),  (\cN +1)^2 \1^{> M-2} \Phi(t)  \right\rangle^{1/2} .
\end{align*}
Now for the first term we use $\1^{\le N} (\cN +1 )\le N+1$ (recall that $\Phi_N(t) \in \cF_+^{\le N}(t)$) and get
\begin{align*} 
\left\langle \1^{> M} \Phi_N(t), (\cN +1)^2 \1^{> M} \Phi_N(t)  \right\rangle \le (N+1)^{2}.
\end{align*}
For the second term, we use $\1^{>M-2} (\cN+1)^2 \le (\cN+1)^{s} (M-1)^{2-s}$ with $s \ge 2$ and then use the moment estimate \eqref{eq:final-proof-moment}. We find that
\begin{align*} 
\left\langle \1^{> M-2} \Phi(t),  (\cN +1)^2 \1^{> M-2} \Phi(t)  \right\rangle  &\le (M-1)^{2-s} \left\langle \Phi(t),  (\cN +1)^s  \Phi(t)  \right\rangle \\
&\le C_{\eps,s} (M-1)^{2-s} N^{2s\eps} [\log(2+t)]^{2s}.
\end{align*}
All this yields   
$$
E_2 \le C_{\eps,s}N^{3\beta+1} M^{1-s/2} N^{s\eps} [\log(2+t)]^{s}.
$$

In summary, from \eqref{eq:final-proof-3} it follows that
\begin{align*} 
\left| \big\langle  \Phi_N(t), (R_{j} +R_j^*)\1^{\le N} \Phi(t) \big\rangle \right| &\le C_\eps  \Big( M N^{(2\beta+\eps-1)/2} + M^{3/2} N^{-1/2} \Big) \\
& + C_{\eps,s}N^{3\beta+1} M^{1-s/2}N^{s\eps} [\log(2+t)]^{s}
\end{align*}
for all $4<M<N-2$ and $s\ge 2$. We can choose $M=N^{3\eps}$ and $s=s(\eps)$ sufficiently large (e.g. $s\ge 6 (1+\beta+\eps)/\eps$) to obtain 
\begin{align} \label{eq:final-proof-2a}
\left| \big\langle  \Phi_N(t), (R_{j} +R_j^*)\1^{\le N} \Phi(t) \big\rangle \right| \le C_\eps \Big( N^{(2\beta+9\eps-1)/2} + N^{-1}(1+t)^{\eps} \Big) .
\end{align}

\medskip

\noindent 
{\bf Step 3.} From \eqref{eq:final-proof-2}, \eqref{eq:final-proof-bH>N-last} and \eqref{eq:final-proof-2a}, we find that 
$$
\partial_t \| \Phi_{N}(t)-\Phi(t)\|^2 \le C_\eps \Big( N^{(2\beta+9\eps-1)/2} + N^{-1} (1+t)^{\eps} \Big).
$$
Integrating over $t$ and using
$$
\| \Phi_N(0)-\Phi(0)\|^2 = \langle \Phi(0), \1^{>N} \Phi(0) \rangle \le N^{-1} \langle \Phi(0), \cN \Phi(0) \rangle \le C_\eps N^{\eps-1}.
$$
we obtain
\begin{align*}
\| \Phi_{N}(t)-\Phi(t)\|^2 &\le C_\eps N^{\eps-1} + C_\eps  \Big( t N^{(2\beta+9\eps-1)/2} + N^{-1} (1+t)^{1+\eps} \Big) \\
& \le C_\eps (1+t)^{1+\eps} N^{(2\beta+9\eps-1)/2}.
\end{align*}
Finally, from \eqref{eq:final-proof-1} we conclude that 
\begin{align} \label{eq:thm-quantitative-estimate}
\| \Psi_{N}(t)- U_N(t)^* \1^{\le N} \Phi(t)\|_{\gH^N}^2 &\le \| \Phi_{N}(t)-\Phi(t)\|^2 \nn \\
& \le C_\eps (1+t)^{1+\eps} N^{(2\beta+9\eps-1)/2}
\end{align}
for all $0<\eps<\min \{1/2,1-2\beta\}$. In the latter estimate we can thus replace $\eps$ by $\eps/9$ and obtain 
\begin{align*}
\| \Psi_{N}(t)- U_N(t)^* \1^{\le N} \Phi(t)\|_{\gH^N}^2 \le C_{\eps} (1+t)^{1+\eps} N^{(2\beta+\eps-1)/2}
\end{align*}
for all $0<\eps<\min \{1/2,1-2\beta\}$ (with the constant $C_\eps$ adjusted appropriately). As we have explained, this estimate holds for all $\eps>0$ because $(1+t)^{1+\eps} N^{(2\beta+\eps-1)/2}\le (1+t)^{1+\eps'} N^{(2\beta+\eps'-1)/2}$ when $\eps'\ge \eps$. This ends the proof of Theorem \ref{thm:main}.
\end{proof}

\end{document}